\documentclass[11pt,a4paper]{article}
  \usepackage{jheppub}
\pdfoutput=1
\usepackage{fancyhdr}
\usepackage{amsthm}
\usepackage{enumerate}
\theoremstyle{definition}
\newtheorem{thm}{Theorem}[section]

\def\beq{\begin{equation}}
\def\eeq{\end{equation}}
\def\bea{\begin{eqnarray}}
\def\eea{\end{eqnarray}}
\def\Sol{\textbf{S}}
\def\SolC{\textbf{S}_\mathbb{C}}
\def\Hil{\mathcal{H}}
\def\HilC{\overline{\mathcal{H}}}
\def\Co{\textrm{C}_0^{\infty}(M)}

\def\Reals{\mathbb{R}}
\def\block{\Box}

\def\hat{\widehat}
\def\ideq{\equiv}               

\newcommand{\ud}{\,\mathrm{d}}

\title{A Distinguished Vacuum State for a Quantum Field in a Curved Spacetime: Formalism, Features, and Cosmology }

\author[a,b]{Niayesh Afshordi,}
\author[a,b]{Siavash Aslanbeigi}
\author[a,c]{and Rafael D. Sorkin}
\affiliation[a]{Perimeter Institute for Theoretical Physics, 31 Caroline St. N., Waterloo, ON, N2L 2Y5, Canada}
\affiliation[b]{Department of Physics and Astronomy, University of Waterloo, Waterloo, ON, N2L 3G1, Canada}
\affiliation[c]{Department of Physics, Syracuse University, Syracuse, NY 13244-1130, U.S.A.}
\emailAdd{nafshordi@perimeterinstitute.ca}
\emailAdd{saslanbeigi@perimeterinstitute.ca}
\emailAdd{rsorkin@perimeterinstitute.ca}

\abstract{
We define a distinguished
``ground state'' or ``vacuum''
for a free scalar quantum field
in a globally hyperbolic region of an arbitrarily curved
spacetime.
Our prescription is motivated by the recent construction
\cite{Johnston,RDS} of a quantum field theory on a background
causal set using only knowledge of the retarded Green's function.  We
generalize that construction to continuum spacetimes and find that it
yields a distinguished
{\it vacuum} or {\it ground state} for a non-interacting,
massive or massless scalar field.
This state is defined for all compact regions and for
many noncompact ones.
In a static spacetime we find that our vacuum coincides with the
usual ground state.
We determine it also
for
a radiation-filled,
spatially homogeneous and isotropic cosmos,
and show that the super-horizon correlations
are approximately the same as those of a thermal state.
Finally, we illustrate the
inherent non-locality of our prescription with the example of a
spacetime which sandwiches a region with curvature in-between flat
initial and final regions.
}

\arxivnumber{1205.1296}

\begin{document}
\maketitle
\flushbottom
\section{Introduction}
The framework known as ``quantum field theory in curved spacetime''
concerns the interaction of quantum fields with gravity,
but only in an asymmetrical sense.
Non-gravitational,
 ``matter'' fields are treated in accord with quantum
principles while
their gravitational ``back reaction''
is either
ignored
entirely or described by a
semiclassical form of the Einstein equations.  Although not a
fundamental theory of nature, this framework has provided us with
profound insights into an eventual
theory of quantum gravity.
Important examples include Hawking radiation by black holes \cite{Hawking},
the Unruh effect \cite{bisognano, Unruh},
and the
generation of Gaussian-distributed random perturbations
in the theory of cosmic inflation \cite{Mukhanov}.
In all these examples a choice of vacuum ---
or at least
a reasonable reference state of the field
--- is of crucial importance.
It therefore seems
unsatisfactory that
as it stands,
quantum field theory lacks a general notion of ``vacuum''
which extends very far beyond flat spacetime.

Formulations of quantum field theory in Minkowski spacetime do provide a
distinguished vacuum, but it rests heavily on a particle interpretation
of the field that is closely tied to the properties of the Fourier
transform and the availability of plane waves.  More abstract treatments
tend to trace the uniqueness of the vacuum to Poincare-invariance, but
that is tied even more closely to flat space.  It is thus unclear how
one might extend the notion of vacuum beyond the case of spacetimes with
a high degree of symmetry.
(Moreover, even a large symmetry-group does not always yield a unique
vacuum without further input.  In de-Sitter for example, one has the one
complex-parameter family of ``$\alpha$-vacua'' \cite{Allen}, all of
which are invariant under the full de-Sitter group.
To single one value
of $\alpha$ out from the rest, one needs
to impose the further condition
that the
two-point function
take the so-called Hadamard form.)

One might even question whether a quantum field theory is well-defined at
all before a vacuum is specified.  What is probably the best
studied mathematical framework for quantum field theory in flat space,
that of the
Wightman axioms, incorporates assertions about the vacuum
among its basic assumptions, and it relies on them
in proving such central results as the PCT and spin-statistics
theorems.
It is therefore noteworthy that the so-called algebraic approach to
quantum field theory has been able to proceed a great distance
without relying on a notion of vacuum, or indeed any unique
representation of the quantum fields at all.
In place of a Poincar{\'e}-invariant vacuum, it has been proposed to
rely on a distinguished {\it class} of states, the so-called Hadamard
states (which are well-suited to renormalization of the stress-tensor by
``point-splitting''), supplemented by an assumption about a
short-distance asymptotic expansion for products of quantum fields,
namely the operator product expansion or ``OPE'' (see for
example \cite{Wald}, \cite{Wald2} and references therein).
If such a ``purely algebraic'' approach were to establish itself more
generally, it might diminish the interest in distinguished ``vacua''
for curved spacetimes.  Conversely, if a reasonable definition of a
preferred vacuum state could be obtained it might remove some of the
motivation for a purely algebraic formulation of quantum field
theory.\footnote
{We suspect that lasting enlightenment about the ``best'' formulation of
  quantum field theory will only arrive together with a solution of the
  problem of quantum gravity, by means of a greater theory within which
  that of quantum field theory in curved spacetime will have to be subsumed.}

Let us remark also that histories-based formulations of quantum
mechanics tend to fuse the concept of state with that of equation of
motion.  This shows up clearly in formulations that start from the
``quantum measure'' \cite{DJS, RDS2, RDS3} or ``decoherence functional'' \cite{Hartle},
neither of which can be defined without furnishing a suitable set of
``initial conditions''.  In this sense, one has no dynamical law at all
before a distinguished ``initial state'' is specified.

At a less formal level, the ability to think in terms of particles
offers an obvious benefit to one's intuition.  And, especially in
relation to cosmology, great interest attaches to the question whether
certain sorts of states can be regarded as ``natural'' to certain
regions of spacetime, a question we return to briefly in section \ref{RadiationEra}.
These, then, are two more reasons why the availability of a
distinguished vacuum could be welcome, whether or not it
is logically necessary to quantum field theory as such.

Moreover, what is logically necessary can change
drastically if one passes from the spacetime continuum to some more
fundamental structure, especially if that structure is discrete.
As we will review later, the entire quantization process --- as usually
conceived --- boils down to selecting an appropriate subspace of the
solution space of the Klein-Gordon equation.  But that way of organizing
the problem seems to break down in the case of a causal set.  There, the
notion of ``approximate solution'' seems to be the best that is
available, and one therefore requires a different starting point.

In \cite{Johnston}, such a starting point was found in
(the discrete analog of)
the retarded
Green function.
On that basis
a complete counterpart of the quantum field theory of a
free scalar field was built up,
and
a unique ``vacuum state'' was derived.
Herein, we generalize that derivation to quantum fields on continuum
spacetimes, showing thereby that there is a sensible way to uniquely
define a vacuum state for a scalar field in any globally hyperbolic
spacetime or region of spacetime.
More precisely, we consider the case of a \textit{free scalar field}
in a \textit{globally hyperbolic} spacetime or region of spacetime,
and in that context we put
forward a definition of
\textit{distinguished vacuum state}
that
applies to all compact regions and to a large class of noncompact regions.

It is thus possible to carry the concept of vacuum far beyond the
confines of Minkowski space by means of definitions we expose in
detail below.  Although, for all of the reasons indicated above, this
possibility is of interest in itself, one naturally wants to know to
what extent, and in what sense, our proposal is ``the right one''?  To
that question, only a sufficient number of particular instances of our
vacuum would seem to be germane.
The examples of Minkowski spacetime
and of globally
static
spacetimes furnish important evidence,
but they contain
little that is new physically.  To judge the ultimate fruitfulness of
our prescription, one should, for example, test it against the behavior
of the ``matter fields'' that one actually
encounters
in the early universe.
In Section \ref{friedmann} we make a start on this kind of test,
beginning with the case of a spatially homegeneous and isotropic
cosmology.

\section{Background}
\label{Background}
In this section, we briefly review the quantization,
along traditional lines,
of a free scalar field on a curved spacetime.
We will be a bit careful with the mathematical technicalities because it will benefit us later.
\footnote{Much of the discussion here will follow that of \cite{Wald} and \cite{Ashtekar}.}
\label{background}
Consider a free, real-valued scalar field $\phi$ on a globally hyperbolic spacetime
($M$, $g_{\mu\nu}$) satisfying the Klein-Gordon equation
with mass-parameter $m\ge0$:
%
%
\beq
 \nabla^{\mu}\nabla_{\mu}\phi(x)+m^2\phi(x)=0 \ ,
 \label{KG}
\eeq
where
$\nabla_\mu$ is the covariant derivative operator on $M$.~\footnote%
{We use signature $(+ - - -)$ and set $\hbar=c=1$.}
%
The condition of global hyperbolicity ensures that \eqref{KG} has a well
posed initial-value formulation (see theorem \ref{main} in Appendix
\ref{theorems}).

Let us review some mathematical structures that are important for both
the classical and quantal theories of a free field.
Consider a
foliation
of $M$ by spacelike Cauchy surfaces $\Sigma_t$, labeled by a time parameter $t$.
Let $\Sol$ be the space of all real $\textrm{C}^{\infty}$ solutions of \eqref{KG}
which induce initial data of compact support on some
(and therefore on every)
$\Sigma_t$.
(This restriction on the
solutions
 is just for convenience, so that various
 mathematical structures are well-defined.)
The retarded and advanced Green's functions $G_{R,A}(x,y)$ associated with \eqref{KG} satisfy
\beq
 (\nabla^{\mu}\nabla_{\mu}+m^2)G_{R,A}(x,y)=-\frac{\delta^4(x-y)}{\sqrt{-g}},
\label{RAG}
\eeq
where $g$ is the determinant of the metric-tensor.
By definition  $G_R(x,y)=0$ unless $x\succ{y}$
(meaning $x$ is inside or on the future lightcone of $y$),
and $G_A(x,y)=0$ unless $x\prec y$.

The so-called Pauli-Jordan
function
is defined as
\beq
 \Delta(x,y)\equiv G_R(x,y)-G_A(x,y) = G_R(x,y) - G_R(y,x)  \ .
 \footnote{See Theorem \ref{main3} of Appendix \ref{theorems}.}
  \
 \label{PJ}
\eeq
From it we define an integral operator $\Delta$:
\beq
 (\Delta f)(x) \equiv \int_M \Delta(x,y) f(y) \ud V_y,
 \label{idelta}
\eeq
where $dV_y=\sqrt{-g(y)}d^4y$ is the metric volume element on $M$,
and we take the domain
of $\Delta$ to be
the space
$\textrm{C}_0^{\infty}(M)$
of
all smooth functions of compact support on $M$.
Since $\Delta f$ is
the difference between
two solutions
of the inhomogeneous Klein-Gordon equation
with the same source $f$, it satisfies the homogeneous Klein-Gordon
equation \eqref{KG}:
\beq
       (\block + m^2) \Delta = 0 \ ,
\eeq
where $\block=\nabla_\mu\nabla^\mu$.
Moreover,
since $f$ has compact support,
$\Delta f$ induces smooth initial data of compact support on all Cauchy surfaces, making
$\Delta$ a map from $\textrm{C}_0^{\infty}(M)$ to $\Sol$.
The operator $\Delta$
(or more generally the corresponding quadratic form)
will be of crucial importance to us.
A symplectic structure
$\Omega:\Sol\times\Sol\rightarrow\Reals$
can be defined on $\Sol$:
\beq
 \Omega(\phi_1,\phi_2)
 \equiv
 \int_{\Sigma_t}
      \left[\phi_1 \nabla_{\mu}\phi_2 - \phi_2 \nabla_{\mu}\phi_1\right]
       dS^\mu \ ,
 \label{symp}
\eeq
where $dS_\mu=n_\mu\sqrt{-h}\ud^3x$
with
 $n^a$  the unit normal to $\Sigma_t$,
and $h$ the determinant
of the induced metric on $\Sigma_t$.
The righthand side of \eqref{symp}
is
well defined
because it is
independent of $t$
for all solutions in $\Sol$.

To pass to the quantum theory, one
introduces operator-valued distributions $\hat{\phi}(x)$ that satisfy
the Klein-Gordon equation, and the canonical commutation relations (CCR)
\beq
  [\hat{\phi}(f),\hat{\phi}(g)]
 =
 i\Omega(\Delta f,\Delta g)
 =
 i\int_M  f(x) \Delta(x,y) g(y)  \ud V_x \ud V_y \ ,
 \label{ccr}
\eeq
for all $f,g\in\textrm{C}_0^{\infty}(M)$, where $\hat{\phi}(f)=\int_M\hat{\phi}(x)f(x)dV_x$. The second
equality follows from theorem \ref{main2} (see Appendix \ref{theorems}).
Equation \eqref{ccr} is typically
expressed
as
\beq
  [\hat{\phi}(x),\hat{\phi}(y)] = i\Delta(x,y).
 \label{CCR}
\eeq
%

To obtain operators in Hilbert space, one requires further
a $*$-representation of these relations.
One typically
works with
irreducible, Fock
representations
constructed
as follows:
 \begin{itemize}
\item Complexify the Klein-Gordon solution space to get $\SolC=\{\phi_1+i\phi_2|\phi_1,\phi_2\in\Sol\}$.
\item Define a map $(\,,\,)_{KG}:\SolC\times\SolC\to\mathbb{C}$ by
  $(\phi_1,\phi_2)_{KG}=i\Omega(\overline{\phi}_1,\phi_2)$,
where the bar denotes complex conjugation.
This map enjoys all the properties of a Hermitian inner product except
that it's not positive definite.
\item Choose any subspace $\Hil\subset\SolC$    with   the following properties:
\begin{itemize}
\item The inner product $(,)_{KG}$ is positive definite on $\Hil$, thus making $\Hil$ into
 a \hbox{(pre-)Hilbert} space over $\mathbb{C}$.
\item $\SolC$ is equal to the span of $\Hil$ and its complex conjugate space $\HilC$.
\item For all $\phi^{+}\in\Hil$ and $\phi^{-}\in\HilC$, we have $(\phi^{+},\phi^{-})_{KG}=0$.
\footnote
{From here on, when we refer to ``a basis $\{\psi_i\}$ of the
 Klein-Gordon solution space'', we mean that $\{\psi_i\}$ is an
 orthonormal basis of $\Hil\subset\SolC$ with the above properties.
 \label{basisFN}}
\end{itemize}
\end{itemize}
The Hilbert space is then taken to be the symmetric Fock space associated with $\Hil$,
the field operators being defined as
\beq
\hat{\phi}(x)=\sum_i\psi_i(x)\hat{a}_i+\overline{\psi}_i(x)\hat{a}^{\dagger}_i,
\label{mode}
\eeq
where $\{\psi_i(x)\}$ is any orthonormal basis of (the Cauchy-completed) $\Hil$
with respect to the inner product $(,)_{KG}$,
and where
$\{\hat{a}_i\}$ are the
annihilation operators associated with $\{\psi_i\}$,
satisfying the usual commutation relations,
$[\hat{a}_i,\hat{a}_j]=0$,
$[\hat{a}_i,\hat{a}^\dagger_j]=\delta_{ij}$.
To show
that \eqref{mode} satisfies the CCR,
write
 $\hat{\phi}(f)
 =\sum_i\big(\int_M\psi_i f\big)\hat{a}_i+\big(\int_M\overline{\psi}_i f\big)\hat{a}^\dagger_i=
 -i\sum_i(\overline{\Delta f},\psi_i)_{KG}\hat{a}_i
       +(\overline{\Delta f},\overline{\psi}_i)_{KG}\hat{a}^\dagger_i$, where in the last equality we have used Theorem \ref{main2}.
Then
\bea
[\hat{\phi}(f),\hat{\phi}(g)]
 &=&-\sum_i(\overline{\Delta f},\psi_i)_{KG}
         (\overline{\Delta g},\overline{\psi}_i)_{KG}
        -(\overline{\Delta f},\overline{\psi}_i)_{KG}
         (\overline{\Delta g},\psi_i)_{KG} \notag \\
 &=& (\overline{\Delta f}|\sum_i |\psi_i)_{KG}(\psi_i|
   -|\overline{\psi}_i)_{KG}(\overline{\psi}_i||\Delta g)\notag\\
 &=& (\overline{\Delta f},\Delta g)_{KG} = i\Omega(\Delta f,\Delta g),
\eea
where we have used the fact that
$\sum_i|\psi_i)_{KG}(\psi_i| - |\overline{\psi}_i)_{KG}(\overline{\psi}_i|$
is the identity operator on $\SolC$
(because $\{\psi_i\}$'s satisfy
$(\psi_i,\psi_j)_{KG}=\delta_{ij}$,
$(\psi_i,\overline{\psi}_j)_{KG}=0$, and
$(\overline{\psi}_i,\overline{\psi}_j)_{KG}=-\delta_{ij}$).
Finally,
the vacuum
is
defined as the state annihilated by all $\hat{a}_i$:
$\hat{a}_i\,|0\!>\,=0$.

The trouble, of course, is that the subspace
$\Hil\subset\SolC$
is not unique.
Even if we limit ourselves to Fock representations,
there are many ways to choose
$\Hil$,
and with each one comes a
different set of operators $\hat{a}_i$ and a different vacuum.

\section{The S-J Vacuum}
%
The kernel $i\Delta(x,y)$ defined by \eqref{PJ} has two basic properties:
\begin{itemize}
\item Antisymmetric --- because $G_A(x,y)=G_R(y,x),$\footnote{See Theorem \ref{main3} of Appendix \ref{theorems}.}
\item Hermitian --- i.e. $\overline{i\Delta(y,x)}=i\Delta(x,y)$.
\end{itemize}
Let $L^2(M)$ denote the Hilbert space of all square-integrable functions
on $M$\footnote%
{$L^2(M) = \{\psi:M\to\mathbb{C} | \int_M |\psi(x)|^2\ud V_x<\infty\}$.}
with the usual inner product
\beq
 <\phi,\psi> \; = \int_{M} \overline{\phi}(x) \psi(x) \ud V_x \ .
 \label{L2}
\eeq
Then, the integral operator $i\Delta$ associated with $i\Delta(x,y)$
[defined in \eqref{idelta}]
will be Hermitian on a subspace of $L^2(M)$: $<\phi,i\Delta\psi> \,=\, <i\Delta\phi,\psi>$
for all $\phi$ and $\psi$ in the domain of $i\Delta$.
\footnote%
{$<\phi,i\Delta\psi>$ = $\int_M\int_M i\Delta(x,y)\overline{\phi}(x)\psi(y)\ud V_x\ud V_y=
 -\int_M\int_M i\Delta(y,x)\overline{\phi}(x)\psi(y)\ud V_x\ud V_y=\int_M\overline{i\Delta\phi(y)}\psi(y)\ud V_y=<i\Delta\phi,\psi>$.}
Strictly speaking, $i\Delta$ is in general
only a densely defined quadratic form in $L^2(M)$,
but in introducing our prescription of a vacuum,
let us at first set aside the functional-analytical
subtleties associated with the domain of $i\Delta$
and assume it to be a self-adjoint operator on $L^2(M)$,
so that $i\Delta(x,y)$ can be ``diagonalized'' in the sense of the spectral theorem.
We can now state our prescription for a distinguished state,
which we call the \textit{S-J vacuum} after the authors of \cite{Johnston,RDS}.
\\\\
\textit{Assuming that $i\Delta$ is selfadjoint,
 a vacuum state $|SJ\!\!>$
 can be defined covariantly via $<\!SJ| \; \hat{\phi}(x) \, \hat{\phi}(y) \, |SJ\!> \, = Pos(i\Delta(x,y))$,
 where $Pos(i\Delta(x,y))$ is the positive spectral projection of $i\Delta$
 provided by the spectral theorem.}
\\\\
Informally speaking, this just means the following: $(i)$ construct $\Delta(x,y)$ by
anti-symmetriz\-ing the retarded Green's function, which is uniquely determined
from the Klein-Gordon equation in any globally hyperbolic spacetime (see theorem
\ref{main} in Appendix \ref{theorems}), $(ii)$ diagonalize $i\Delta(x,y)$
in the $L^2$ norm, $(iii)$ take its `positive part' to be the two-point function. The spectral theorem
gives precise mathematical sense to this prescription so long as $i\Delta$ is self-adjoint.

Of course a two-point function is not yet a full characterization of a
state, but
it becomes so if we take the state to be ``gaussian'' by appropriately
expressing the $n$-point functions in terms of the two-point function
(i.e. by means of the Wick rule).  It then will follow,
given (\ref{mode}) and (\ref{W}) below,
that the Gel'fand-Naimark-Segal
(GNS)
representation associated with
$|SJ\rangle$ will be a Fock representations with
$|SJ\rangle$ playing the role of vacuum.


Let us consider first the special case
where $i\Delta(x,y)\in L^2(M\times M)$,
\footnote{i.e. $\int_M\int_M|i\Delta(x,y)|^2\ud V_x \ud V_y< \infty$.}
as is the case in a bounded\footnote%
{Bounded = having compact closure.}
globally hyperbolic region of 1+1 dimensional Minkowski space, for
instance. \footnote%
{For a massive scalar field in 1+1 dimensional Minkowski space with metric $ds^2=dt^2-dx^2$,
 $G_R(t,x;t',x')=G(t-t';x-x')$ where
 $G(t,x)=\theta(t)\theta(s^2)\frac{1}{2}J_0(ms)$,
 $J_0$ is a Bessel function of the first kind,
 and $s^2=t^2-x^2$\cite{Johnston}.}
Then
$i\Delta$ becomes a so-called Hilbert-Schmidt integral operator and the following version of the spectral theorem
applies \cite{RS}: there exists an orthonormal
basis $\{T_k\}_{k=1}^{\infty}$ of $L^2(M)$ consisting of
eigenfunctions of $i\Delta$ which satisfy $i\Delta T_k=\lambda_kT_k$
with $\lambda_k\in\mathbb{R}$.
Using this theorem and the fact that $\Delta(x,y)$ itself is real,
we deduce $i\Delta \overline{T}_k=-\lambda_k\overline{T}_k$,
which in turn makes it possible to split $i\Delta(x,y)$ into a
{\it positive} and a {\it negative} part:
$$
 i\Delta(x,y)=\sum_{k=1}^{\infty}\lambda_k\big\{T_k(x)\overline{T}_k(y)-\overline{T}_k(x)T_k(y)\big\}
$$
(taking $\lambda_k>0$ now).
In this case, our prescription can be expressed as
\beq
 W(x,y) \ideq <SJ|\hat{\phi}(x)\hat{\phi}(y)|SJ> = \sum_{k=1}^{\infty}\lambda_kT_k(x)\overline{T}_k(y) \ .
 \label{W}
\eeq
This is equivalent to introducing field operators
$\hat{\phi}(x)=\sum_k\sqrt{\lambda_k}\big\{T_k(x)\hat{a}_k+\overline{T}_k(x)\hat{a}_k^\dagger\big\}$,
because
$(i)$ eigenfunctions of $i\Delta$ with $\lambda\not=0$
necessarily
satisfy the Klein-Gordon equation\footnote%
{$i\Delta f = \lambda f \implies (\block+m^2)f=(i/\lambda)(\block+m^2)\Delta f = 0$
 since $(\block+m^2)\Delta=0$.}
and
$(ii)$ the commutation relations are trivially satisfied.
\footnote{Note that in this case there is no need to `smear out' the field operators
 with smooth test-functions of compact support, because $i\Delta(x,y)$ is a completely
 well-defined function (in the $L^2$ sense).}
The S-J vacuum is then the state in Fock space that is annihilated by all $\hat{a}_k$.

Now let us turn to  3+1 dimensions, where $i\Delta$ is no longer Hilbert-Schmidt.
Nonetheless, $i\Delta(x,y)$ is still
a distribution and,
at least within Minkowski space,
it defines a self-adjoint operator
if $M$ is bounded
(see section \ref{com1} of Appendix \ref{comments}).
Thanks to the spectral theorem, our prescription then retains a precise mathematical
sense, and
it is not too far-fetched to assume that this continues to hold for curved spacetimes,
because curvature should not change the singularity structure of $i\Delta(x,y)$ too drastically.
\footnote{Fewster and Verch have now established rigorously that
  our proposal is well-defined on all bounded globally hyperbolic
  spacetimes and that the S-J vacuum is a ``pure quasi-free state''
  \cite{FV}.}

Although selfadjointness might seem to be merely a technical issue, it highlights the fact
that the S-J vacuum depends on a choice of (globally hyperbolic) spacetime region.
Indeed, as we have just seen, our prescription is not guaranteed to be well defined unless
one chooses a region that is bounded, both spatially and temporally.  Thus arise two
questions:  To what extent does our prescription
depend on boundary condtions,
and to what extent does it
remain well-defined in unbounded spacetimes?

In answering the first question, one must distinguish between spatial boundaries (also
referred to as timelike boundaries) and temporal ones.  Spatial boundaries are familiar to
us from putting fields in a box, Casimir effect, etc; and they seem unproblematic.  When
they are present the S-J vacuum will be sensitive to one's choice of boundary conditions,
because the retarded and advanced Green functions be depend on them.  But this is as it
should be since the physics genuinely depends on the boundary conditions.  We will also
consider below regions which are unbounded spatially, but no special difficulties will
arise from that feature.

The case of a temporal boundary (spacelike or null) is less familiar.  The first thing to
notice is that
{\it boundary conditions are neither needed nor possible} in this case,
since the region is (by assumption) globally hyperbolic.
Mathematically, this very satisfactory feature stems from the fact that
$\Delta$ is not a differential operator but an integral one.
Nevertheless, one must bear in mind that the ``ground state'' one ends
up with, does depend on the region with which one begins.
In itself, this dependence on the region
merely expresses the nonlocal character of our definition.
One might for example be interested in which ``vacuum'' would be appropriate to an early
stage of expansion of the cosmos, and one would not want in that case to apply our
definition to the full spacetime, including its whole future development.
However, one might also want to apply the definition to unbounded spacetimes like
Minkowski space, and in such cases one needs to worry about dependence on an eventual infrared
cutoff.
If the metric is static, for instance, why should time play any
role in what the vacuum state looks like?

In dealing with such instances, it is always
possible to work first
 in a truncated spacetime,
and later send the temporal boundary to infinity.
In section \ref{com3} of Appendix \ref{comments}, we apply this
method to the simple harmonic oscillator and show that it
succeeds
in the sense that the resulting S-J vacuum is the minimum
energy state of the Hamiltonian.
As we will later demonstrate, this continues to be true for all static spacetimes.
Another example
of such a calculation can be found in \cite{Yasaman}, where the spectrum of
$i\Delta$ is computed in a 1+1 dimensional causal-diamond, and it is
found that
(up to the usual infrared ambiguities that affect massless scalars in 2d)
the resulting two-point function has the correct (i.e. Minkowski) limiting
behaviour as the boundaries of the diamond tend to infinity.
However, there are also cases where taking a temporal cutoff to infinity is an ill-defined
procedure.
In \cite{deSitter}, it
is shown that the so-called Poincar{\'e} patch of de Sitter space provides an example of
such a case.
(We suspect that this kind of ambiguity can be understood intuitively as
 the failure of $i\Delta$ to admit a selfadjoint extension which is
 unique.  However we don't know how to pose such a question properly,
 because $\Delta$ is densely defined only as a quadratic form on
$L^2(M)$, not as an operator from $L^2(M)$ to itself.)
\\\\
Diagonalizing
$i\Delta$ lies at the heart of our prescription,
at least in practice.
In this section, we will attempt to frame this problem as
generally
as possible.
In subsequent sections, we will deal with more concrete examples.
Let $\{u_k(x)\}$ be a basis for the Klein-Gordon solution space.
\footnote{See footnote \ref{basisFN} on the definition of ``basis''.} 
Expanding
the field operator in terms of these modes as
$\hat{\phi}(x)=\sum_ku_k(x)\hat{a}_k+\overline{u}_k(x)\hat{a}_k^\dagger$
and computing the commutator yields,
as we have seen,
\beq
 [\hat{\phi}(x),\hat{\phi}(y)]=\sum_ku_k(x)\overline{u}_k(y)-u_k(y)\overline{u}_k(x).
\eeq
It then follows from the CCR
\eqref{CCR}
that the integral-kernel $i\Delta(x,x')$
takes the form,%
\footnote{In Section \ref{com2} of Appendix \ref{comments} we confirm
 that the two sides of this equation are equal when integrated against
 an arbitrary test function.}
\beq
 i\Delta(x,y)=\sum_{k}u_k(x)\overline{u}_k(y)-\overline{u}_k(x)u_k(y).%
 \label{E1}
\eeq
Of course the choice of
the $u_k(x)$
in this expansion is not unique.
Another set of  modefunctions
$g_k(x)=\sum_{n}\alpha_{kn}u_n(x)+\beta_{kn}\bar{u}_n(x)$
will give the same commutator,
so long as the following normalization conditions are met:
$\sum_k\alpha_{ik}\bar{\alpha}_{jk}-\beta_{ik}\bar{\beta}_{jk}=\delta_{ij}$ and
$\sum_k\alpha_{ik}\beta_{jk}-\beta_{ik}\alpha_{jk}=0$,
this being nothing but a Bogoliubov transformation.

By means of such a transformation,
we can find orthonormal eigenfunctions $T_n(x)$
of $i\Delta$  (with corresponding eigenvalues $\lambda_n$)
starting from any convenient basis $\{u_k(x)\}$ of the Klein-Gordon solution space.
Requiring $T_n(x)$'s to be
eigenfunctions of $i\Delta$ and using \eqref{E1}:
\beq
 i\Delta \, T_n(x) = \lambda_nT_n(x)=\sum_k\langle u_k,T_n\rangle u_k(x)-\langle\bar{u}_k,T_n\rangle\bar{u}_k(x) \ .
\eeq
For notational simplicity, let
$\alpha_{nk}
\equiv
\frac{<u_k,T_n>}{\lambda_n}$
and
$\beta_{nk}
\equiv
-\frac{<\bar{u}_k,T_n>}{\lambda_n}$ so that
$T_n(x)=\sum_k\alpha_{nk}u_k(x)+\beta_{nk}\bar{u}_k(x)$.
These coefficients then satisfy:
\bea
\alpha_{nk}&=&\frac{1}{\lambda_n}\sum_m\alpha_{nm}\langle u_k,u_m\rangle+\beta_{nm}\langle u_k,\bar{u}_m\rangle \label{alpha} \\
\beta_{nk}&=&\frac{-1}{\lambda_n}\sum_m\alpha_{nm}\langle\bar{u}_k,u_m\rangle+\beta_{nm}\langle\bar{u}_k,\bar{u}_m\rangle. \label{beta}
\eea
Requiring these eigenfunctions to be orthonormal
($\forall$ $n,m $)
yields
\bea
\langle T_n,\bar{T}_m\rangle=0 & \Longleftrightarrow & \sum_k\alpha_{nk}\beta_{mk}-\alpha_{mk}\beta_{nk}=0
\label{norm1}\\
\langle T_n,T_m\rangle=\delta_{nm} &\Longleftrightarrow & \sum_k\bar{\alpha}_{nk}\alpha_{mk}-\bar{\beta}_{nk}\beta_{mk}=\frac{\delta_{nm}}{\lambda_m}.
\label{norm2}
\eea
Then, diagonalizing $i\Delta$ boils down to finding $\alpha_{nk}$ and $\beta_{nk}$ by solving these four equations.
(We have not addressed the issue of convergence in the sums appearing
above. In fact, if it turned out that the $\{T_n(x)\}$ and $\{u_n(x)\}$
induce unitarily inequivalent representations of CCR, the above sums
would not converge.)

It is not obvious how this can be done generically.
To proceed, let us simplify
this calculation by \textit{assuming} that
there are modefunctions $\{u_k(x)\}$ that satisfy
\bea
\langle u_k,u_m\rangle&=&\langle u_k,u_k\rangle\delta_{km} \label{assump1} \\
\langle u_k,\bar{u}_m\rangle&=&\langle u_k,\bar{u}_{-k}\rangle\delta_{k,-m}. \label{assump2}
\eea
The notation used here is as follows: for every
 $u_n\in\{u_k\}$, there is one (and only one) member of the
 complex conjugate set $\bar{u}_{-n}\in\{\bar{u}_k\}$,
for which $\langle u_n,\bar{u}_{-n}\rangle$ can be non-zero.
Also, we denote the complex conjugate of $\bar{u}_{-n}$ by $u_{-n}$.

Our
assumption is motivated by spacetimes for which these modefunctions are plane waves.
Under this assumption,  $\alpha_{nk}$ and $\beta_{nk}$ can be found:
\bea
\alpha_{nk}&=&\left[\frac{\lambda_{-n}+\langle u_n,u_n\rangle}{\lambda_n(\lambda_n+\lambda_{-n})}\right]^{\frac{1}{2}}\delta_{nk} \\
\beta_{nk}&=&-\left[\frac{\langle u_n,u_n\rangle-\lambda_{n}}{\lambda_n(\lambda_n+\lambda_{-n})}\right]^{\frac{1}{2}}
e^{-i\text{Arg}(<u_n,\bar{u}_{-n}>)}\delta_{n,-k},
\eea
where $\lambda_n$ is given by
\beq
2\lambda_n=\langle u_n,u_n\rangle-\langle u_{-n},u_{-n}\rangle+
\Big[\big(\langle u_n,u_n\rangle+\langle u_{-n},u_{-n}\rangle\big)^2-4|\langle u_n,\bar{u}_{-n}\rangle|^2\Big]^{\frac{1}{2}}.
\label{eigenV}
\eeq By direct substitution, it can be verified that these indeed solve \eqref{alpha}, \eqref{beta}, \eqref{norm1}
and \eqref{norm2}. Finally, the vacuum state modefunctions as picked by our prescription
$\mathcal{U}_n^{SJ}(x)\equiv\sqrt{\lambda_n}T_n(x)$ take the form
\bea
 \sqrt{2} \mathcal{U}_n^{SJ}(x)&=&
 \left[1+\frac{1}{\sqrt{1-\frac{4|<u_n,\bar{u}_{-n}>|^2}{\left(<u_n,u_n>+<u_{-n},u_{-n}>\right)^2}}}\right]^{\frac{1}{2}} u_n(x) \notag\\
  &-& \left[-1+\frac{1}{\sqrt{1-\frac{4|<u_n,\bar{u}_{-n}>|^2}{\left(<u_n,u_n>+<u_{-n},u_{-n}>\right)^2}}}\right]^{\frac{1}{2}}e^{-i\text{Arg}(<u_n,\bar{u}_{-n}>)} \bar{u}_{-n}(x).
\label{VAC}
\eea
Looking closely at \eqref{VAC}, we see that our prescription picks out a particular basis of the Klein-Gordan
solution space that satisfies \eqref{assump1}, \eqref{assump2}, and
\beq
\frac{|<u_n,\bar{u}_{-n}>|}{<u_n,u_n>+<u_{-n},u_{-n}>}=0.
\label{zero}
\eeq
In a
bounded
region of spacetime where all inner products are finite, \eqref{zero}
implies that
\hbox{$<u_n,\bar{u}_{-n}>=0$}.
That this is a unique choice can be shown on more general grounds (see section \ref{com4} of Appendix \ref{comments}).
Even in an unbounded spacetime, where the inner-products might diverge,
we may deem \eqref{zero} to be satisfied
so long as the denominator diverges more strongly than the numerator.
This is very similar to the case of the
simple harmonic oscillator (see section \ref{com3} of Appendix \ref{comments}).
However, it might be the case that
$\frac{2|<u_n,\bar{u}_{-n}>|}{<u_n,u_n>+<u_{-n},u_{-n}>}\to1$
as the limit is taken to
infinity.
In such a case,
the prescription (\ref{VAC}) fails.
\section{Consistency with Known Vacua: Static Spacetimes}
In static spacetimes,
i.e. spacetimes that admit an everywhere time-like and hypersurface-orthogonal
Killing vector $k^{\mu}$, a natural choice of vacuum modefunctions
exists,
namely the
solutions of the Klein-Gordon equation which are
purely positive frequency with respect to the Killing time.
(The corresponding vacuum state minimizes the Hamiltonian).
For a massive scalar field $\psi(x)$ in a static spacetime the Klein-Gordon equation reads
\beq
\left(\frac{\partial^2}{\partial t^2}+K\right)\psi(x)=0,
\eeq where $K=\frac{1}{g^{00}(\vec{x})}\Big\{\frac{1}{\sqrt{-g(\vec{x})}}\partial_i\big[g^{ij}(\vec{x})\sqrt{-g(\vec{x})}\partial_j\big]+m^2\Big\}$
is a purely spatial differential operator.
Let $L^2(\Sigma)$ denote the Hilbert space of all $L^2$ functions on the spatial domain $\Sigma$ with inner product
\beq
<f,g>_S=\int_{\Sigma}\bar{f}(\vec{x})g(\vec{x})\frac{\sqrt{-h}}{|k|}\ud^3\vec{x},
\eeq  where $|k|=\sqrt{k_{\mu}k^{\mu}}=\sqrt{g_{00}}$, and $h$ is the determinant of the induced metric on $\Sigma$.
 Let us assume that $K$ is a self-adjoint and strictly positive operator on $L^2(\Sigma)$
so that it has a well-defined positive spectrum: $(K u_{k})(\vec{x})=\omega(k)^2u_{k}(\vec{x})$.
Then, it is always possible to find complex solutions of the Klein-Gordon equation
of the form $\psi_{k}(x)=\frac{e^{-i\omega(k) t}}{\sqrt{2\omega(k)}}u_{k}(\vec{x})$  ($\omega>0$) that satisfy
\beq
L_k\psi_k=k^{\mu}\nabla_{\mu}\psi_k=-i\omega \psi_k,
\label{static}
\eeq where $L_k\psi_n$ denotes the Lie derivative of $\psi_n$ along the Killing vector $k^{\mu}$.
Moreover, in static spacetimes, we can choose  $n^{\mu}=\frac{k^{\mu}}{|k|}$  as the unit normal to $\Sigma$.
It then follows that $\{\psi_k(x)\}$ form an orthonormal basis of the Klein-Gordon solution space
so long as $<u_k,u_{k'}>_S=\delta_{kk'}$:
\bea
(\psi_k,\psi_{k'})_{KG}&=&\frac{\omega(k)+\omega(k')}{\sqrt{2\omega(k)}\sqrt{2\omega(k')}}
e^{it(\omega(k)-\omega(k'))}<u_k,u_{k'}>_S=\delta_{kk'} \\
(\bar{\psi}_k,\psi_{k'})_{KG}&=&\frac{\omega(k')-\omega(k)}{\sqrt{2\omega(k)}\sqrt{2\omega(k')}}
e^{-it(\omega(k)+\omega(k'))}<\bar{u}_k,u_{k'}>_S=0.
\eea
\footnote{Note that $(\bar{\psi}_k,\psi_{k'})$ vanishes identically because $<\bar{u}_k,u_{k'}>_S$ is only nonzero when
$\bar{u}_k$ and $u_{k'}$ both have the same
eigenvalue, but in that case $\omega(k')-\omega(k)=0$.}
Let us now turn to the S-J prescription. 
It can be verified that
\bea
<\psi_k,\psi_{k'}>&=&<\sqrt{\frac{\pi g_{00}}{\omega(k)}}u_k,\sqrt{\frac{\pi g_{00}}{\omega(k')}}u_{k'}>_S \delta(\omega(k)-\omega(k'))\\
<\psi_k,\bar{\psi}_{k'}>&=&<\sqrt{\frac{\pi g_{00}}{\omega(k)}}u_k,\sqrt{\frac{\pi g_{00}}{\omega(k')}}\bar{u}_{k'}>_S\delta(\omega(k)+\omega(k'))=0,
\eea
where we have taken $t\in(-\infty,\infty)$. This implies that $i\Delta$ does not mix the positive and negative frequency modefunctions, which immediately follows from \eqref{E1}:
\beq
(i\Delta\psi_k)(x)=\sum_{k'}<\psi_{k'},\psi_{k}>\psi_{k'}(x).
\label{stat}
\eeq This equation can be viewed as a matrix
$\overline{<\psi_{k},\psi_{k'}>}$
multiplying a vector $\psi_{k'}(x)$, where $\overline{<\psi_{k},\psi_{k'}>}$ is \textit{Hermitian} and can be diagonalized.
As a result, it is always possible to find eigenfunctions of $i\Delta$
as linear combinations of \textit{purely positive frequency modes}, which in turn implies that the S-J vacuum state is the same
as the vacuum state defined by $\{\psi_k\}$'s. We have overlooked some potential technical difficulties in this argument. For instance,
it might be the case that the inner products $<\sqrt{g_{00}}u_k,\sqrt{g_{00}}u_{k'}>_S$ are infinite because the spatial domain is not compact.
However, because in this case
the eigenfunctions of $i\Delta$ are only linear combinations of purely positive frequency modes, the resulting two-point function will \textit{not}
depend on the regularization scheme, as it ought to be equal to $<SJ|\hat{\phi}(x)\hat{\phi}(y)|SJ>=\sum_k\psi_k(x)\bar{\psi}_k(y)$.
\footnote{The same argument applies to the $\delta(0)$'s that appear.}
%

Consider, for example,
the 1+1 dimensional Rindler wedge:
\beq
ds^2=e^{2a\xi}(d\eta^2-d\xi^2),
\eeq
where $-\infty<\eta,\xi<\infty$.
With $m=0$,
the normalized positive frequency modes  with respect to the Killing vector $\partial_{\eta}$ are
\beq
 \psi_k=\frac{1}{\sqrt{4\pi|k|}}e^{-i|k|\eta+ik\xi}.
\eeq
These
 modefunctions form an orthonormal basis of the Klein-Gordon solution
space:
$(\psi_k,      \psi_{k'})_{KG} = \delta(k-k')$ and
$(\psi_k,\bar{\psi}_{k'})_{KG} = 0$.
Moreover,
\bea
 <\psi_k,\psi_{k'}>&=&\frac{\delta(|k|-|k'|)}{\sqrt{4|k||k'|}}\int_{-\Lambda}^{\Lambda}e^{i\xi(k'-k)}e^{2a\xi}\ud\xi \\
 <\bar{\psi}_k,\psi_{k'}>&=&\frac{\delta(|k|+|k'|)}{\sqrt{4|k||k'|}}\int_{-\Lambda}^{\Lambda}e^{i\xi(k'+k)}e^{2a\xi}\ud\xi=0,
\eea
where we have regulated the spatial integrals with a cut-off $\Lambda$. In this case, all inner products
vanish except for $<\psi_k,\psi_{k}>$ and $<\psi_k,\psi_{-k}>$ and it follows from \eqref{E1} that
$(i\Delta\psi_k)(x)=<\psi_k,\psi_k>\psi_k(x)+<\psi_{-k},\psi_k>\psi_{-k}(x)$. Then, it can be verified that orthonormal eigenfunctions $T_k(x)$
of $i\Delta$ (with eigenvalue $\lambda_k$) take the form
\beq
 T_k(x)=\frac{1}{\sqrt{2\lambda_k}}\Big(\psi_k(x)+e^{i\text{Arg}<\psi_{-k},\psi_k>}\psi_{-k}(x)\Big)
\eeq
where
\beq
 \lambda_k=\frac{1}{<\psi_{k},\psi_k>+|<\psi_{-k},\psi_k>|}.
\eeq
Therefore, the S-J vacuum modefunctions are
\beq
 \mathcal{U}_n^{SJ}(x)\equiv\sqrt{\lambda_n}T_n(x)=\frac{1}{\sqrt{2}}\Big(\psi_k(x)+e^{i\text{Arg}<\psi_{-k},\psi_k>}\psi_{-k}(x)\Big).
\eeq
In this expression, $\text{Arg}<\psi_{-k},\psi_k>$ depends on the cut-off $\Lambda$ used
to regulate the spatial integrals. However, as previously argued, this
makes no difference
because the two-point function is independent of $\Lambda$.

It is worth noting
that
the foregoing
 analysis does \textit{not} apply to stationary spacetimes that are not
 static,
including cases where the Killing vector under consideration is not everhwhere timelike.
It would
be particularly interesting to investigate the
S-J vacuum
in the
spacetime
of a rotating star with an ergo-region.

\section{Application to Non-stationary Spacetimes}
\label{friedmann}
Quantum field theory on time-dependent backgrounds is of particular
importance because the universe we live in is not static. The choice of
vacuum
 in such cases is not at all trivial.
For example, in a
Friedmann-Lema\^{\i}tre-Robertson-Walker (FLRW) spacetime, one choice of
instantaneous
vacuum
is obtained by minimizing the
Hamiltonian at
the given
instant in time. This might seem like a natural
generalization from static spacetimes, but as is well-known by now,
it suffers from severe physical problems like infinite particle
production \cite{Fulling}.  In this section, we will work out the
S-J vacuum state in a spatially-flat FLRW spacetime for
some specific cases.
(In \cite{DHP}, a prescription was introduced for singling out a
Hadamard state in the case of a spatially-flat FLRW universe with a
scale factor that is either exponential or a power-law at early
times. We defer a comparison between the SJ state and that introduced in
\cite{DHP} to future work.)
The metric reads
\beq
 ds^2=a(\eta)^2[d\eta^2-d\vec{x}^2],
\eeq
where $a(\eta)$ and $\eta$ are the scale factor and conformal time, respectively.
A basis\footnote%
{See footnote
\ref{basisFN}
on the definition of ``basis''.} 
 $\{u_k\}$ for the Klein-Gordon solution space may be constructed as
\beq
 u_k(x)=\frac{e^{i\vec{k}\cdot\vec{x}}}{\sqrt{(2\pi)^3}}\frac{g_k(\eta)}{a(\eta)}
 \eeq where $k$ is the comoving Fourier wavenumber and $g_k(\eta)$
satisfies\footnote{$k=|\vec{k}|$ and $a'=\frac{da(\eta)}{d\eta}$.}
\beq
 \left(\partial_{\eta}^2+k^2+m^2a^2-\frac{a^{''}}{a}\right)g_k(\eta)=0
 \label{EOM2}
\eeq
\beq
g_k(\eta)\partial_{\eta}\bar{g}_{k}(\eta)-\bar{g}_k(\eta)\partial_{\eta}g_{k}(\eta)=i.
\label{WR2}
\eeq
Satisfying \eqref{EOM2} is eqivalent to satisfying the Klein-Gordan equation,
while the
normalization of the
Wronskian in \eqref{WR2} is equivalent to
$(u_k,u_{k'})_{KG}=\delta^3(\vec{k}-\vec{k'})$.
We have also
 $(u_k,\bar{u}_{k'})_{KG}=0$.
Moreover,
\bea
 <u_k,u_{k'}>&=&<g_k,g_k>_{\eta}\delta^3(\vec{k}-\vec{k'})\\
 <u_k,\bar{u}_{k'}>&=&<g_k,\bar{g}_{k}>_{\eta}\delta^3(\vec{k}+\vec{k'}),
\eea
where $<,>_{\eta}$ is defined by
\beq
 <f(\eta),g(\eta)>_{\eta}\equiv\int_0^{\Lambda}\bar{f}(\eta)g(\eta)a^2(\eta)\ud\eta.
\eeq
As usual, we have regulated the integral with a cutoff $\Lambda$, which
will be taken to infinity after the eigenfunctions of $i\Delta$ are found.
We can now use \eqref{VAC} to compute the spectrum of $i\Delta$:
\begin{align}
 \mathcal{U}_k^{SJ}(x)=
 \frac{e^{i\vec{k}\cdot\vec{x}}}{\sqrt{2(2\pi)^3}a(\eta)}\left\{\sqrt{1+\frac{1}{\sqrt{1-\frac{|<g_k,\bar{g}_k>_{\eta}|^2}{<g_k,g_k>_{\eta}^2}}}}g_k(\eta)-\sqrt{-1+\frac{1}{\sqrt{1-\frac{|<g_k,\bar{g}_k>_{\eta}|^2}{<g_k,g_k>_{\eta}^2}}}}e^{i\text{Arg}<\bar{g}_k,g_k>} \bar{g}_k(\eta)\right\}.
\label{VAC2}
\end{align}
\subsection{Massless field in the radiation era}
In the radiation era $a\propto\eta$.
When $m=0$, $g_k(\eta)=\frac{1}{\sqrt{2k}}e^{-ik\eta}$ satisfies
both \eqref{EOM2} and \eqref{WR2},
and
we have
also
\beq
\frac{|<g_k,\bar{g}_k>_{\eta}|}{<g_k,g_k>_{\eta}}=
\lim_{\Lambda\to\infty}\frac{|\int_0^{\Lambda}e^{2ik\eta}\eta^2\ud\eta|}{\int_0^{\Lambda}\eta^2\ud\eta}=0.
\eeq
Putting this back in \eqref{VAC2}, our prescription picks out the modefunctions:
\beq
\mathcal{U}_k^{SJ}(x)=\frac{1}{\sqrt{(2\pi)^32k}a(\eta)}e^{-i(k\eta-\vec{k}\cdot\vec{x})}.
\eeq
These are the so-called adiabatic-vacuum modefunctions
(for which an exact
expression
 exists in the case of a massless scalar field in a radiation dominated cosmos) \cite{BD}.

%
\subsection{Massive field in the radiation era}
\label{RadiationEra}
For a massive free scalar field in the radiation era, \eqref{EOM2} can still be solved analytically. Let
$z=i\tilde{m}\eta^2$, where $\tilde{m}=\alpha m$ and $\alpha$ is
a constant defined through $a(\eta)=\alpha\eta$. Furthermore, define a function $F$ by
$g_k(\eta)\equiv \frac{F(i\tilde{m}\eta^2)}{\sqrt{\eta}}$. With these
definitions, \eqref{EOM2} becomes
\beq
\frac{\partial^2F}{\partial z^2}+\big(-\frac{1}{4}-\frac{ik^2}{4\tilde{m}z}+\frac{3}{16 z^2}\big)F=0.
\eeq
This
equation
has two independent solutions $W_{\lambda,\mu}(z)$ and $W_{-\lambda,\mu}(-z)$
(called Whittaker functions)
with $\lambda=\frac{-ik^2}{4\tilde{m}}$ and $\mu=\frac{1}{4}$
(see e.g. \cite{Olver}).
In our case, these two functions are complex conjugates of one another.
Using the properties of Whittaker functions, \footnote{$\mathcal{W}\Big\{W_{k,\mu}(z),W_{-k,\mu}(e^{\pm i\pi}z)\Big\}=e^{\mp ik\pi}$
where
$\mathcal{W}$ is the Wronskian \cite{Olver}. }
it can be shown that $g_k$ satisfies the Wronskian condition
\eqref{WR2} with the normalization:
\beq
g_k(\eta)=\frac{e^{\frac{-\pi k^2}{8\tilde{m}}}}{\sqrt{2m}a(\eta)^{1/2}}W_{\lambda,\mu}(i\tilde{m}\eta^2).
\label{RM}
\eeq
\\
As we will soon show, $|g_k(\eta)|$ is constant for small $\eta$, which
means all inner products are
finite
in this
region.
Divergences
arise for large $\eta$,
though.
In this regime, $W_{\lambda,\mu}(z)\longrightarrow e^{-\frac{1}{2}z}z^{\lambda}$ and plugging this into \eqref{RM} we find:
\beq
g_k(\eta)\longrightarrow \frac{1}{\sqrt{2m}a(\eta)^{1/2}}e^{-\frac{i}{2}[\tilde{m}\eta^2+\frac{k^2}{2\tilde{m}}\ln(\tilde{m}\eta^2)]}.
\eeq
\\
Just as before,
it can be checked that $\lim_{\Lambda\to\infty}\frac{|<g_k,\bar{g}_k>_{\eta}|}{<g_k,g_k>_{\eta}}=0$,
\footnote{This is because $|g_k|^2a(\eta)^2$ diverges quadratically in $\eta$, while $g_k^2a(\eta)^2$ oscillates$\sim e^{-i\tilde{m}\eta^2}$.}
whence
our prescription picks out the modefunctions
\beq
\mathcal{U}_k^{SJ}(x)=\frac{e^{i\vec{k}\cdot\vec{x}}e^{\frac{-\pi k^2}{8\tilde{m}}}}{\sqrt{2m(2\pi)^3}a(\eta)^{3/2}}W_{\frac{-ik^2}{4\tilde{m}},\frac{1}{4}}(i\tilde{m}\eta^2).
\eeq
The corresponding two-point function is:
\beq
\langle SJ|\hat{\phi}(x)\hat{\phi}(x')|SJ\rangle=\int\frac{\ud^3\vec{k}}{(2\pi)^3}\frac{e^{i\vec{k}\cdot(\vec{x}-\vec{x}')-\frac{\pi k^2}{4\tilde{m}}}}{2m a(\eta)^{3/2}a(\eta')^{3/2}}
  W_{\frac{-ik^2}{4\tilde{m}},\frac{1}{4}}(i\tilde{m}\eta^2)W_{\frac{ik^2}{4\tilde{m}},\frac{1}{4}}(-i\tilde{m}\eta'^2) \ .
\eeq

It is reasonable to ask whether this vacuum state could potentially have observable effects.
One way of approaching this problem is to calculate the response rate of a comoving detector, such as
the Unruh-Dewitt detector, when the field is in the S-J vacuum state.
Even more ambitiously, one could (in principle) derive the S-J vacuum
for a general scale-factor $a(\eta)$ (with reasonable boundary conditions),
and study
its
 back-reaction on the underlying geometry
via the renormalized
stress-energy-momentum tensor.
These computations are fairly cumbersome and we defer a
detailed treatment to future studies.

In order to gain \textit{some} intuition, however,
we will compute
$\langle SJ|\hat{\rho}|SJ\rangle\equiv\langle SJ|\hat{T}^{0}_{\phantom{0}0}|SJ\rangle=a^{-2}\langle SJ|\hat{T}_{00}|SJ\rangle$,
where
the expectation value of the
(un-renormalized)
energy momentum tensor $\hat{T}_{\mu\nu}$ takes the form \cite{BD}:
\bea
 \langle SJ|\hat{T}_{\mu\nu}|SJ\rangle&=&\int\ud^3\vec{k}T_{\mu\nu}\left(\mathcal{U}_k^{SJ},\bar{\mathcal{U}}_k^{SJ}\right),\\
 T_{\mu\nu}(\phi,\psi)
 &\ideq&
 \nabla_{\mu}\phi\nabla_{\nu}\psi-\frac{1}{2}g_{\mu\nu}\left[\nabla^{\alpha}\phi\nabla_{\alpha}\psi-m^2\phi\psi\right].
\eea
It can be checked that
\bea
T_{00}\left(\mathcal{U}_k^{SJ},\bar{\mathcal{U}}_k^{SJ}\right)&=&\frac{1}{2}\left[\left|\partial_{\eta}\mathcal{U}_k^{SJ}\right|^2+\left|\vec{\nabla}\mathcal{U}_k^{SJ}\right|^2+m^2a^2\left|\mathcal{U}_k^{SJ}\right|^2\right]\\
&=&\frac{1}{(2\pi)^32a^2}\left[\left|g_k'-\frac{a'}{a}g_k\right|^2+(k^2+m^2a^2)|g_k|^2\right],
\label{tmunu}
\eea
where $'$ denotes differentiation with respect to $\eta$. Using the expression for $g_{k}(\eta)$ given by \eqref{RM}, it follows that
\beq
\langle SJ|\hat{\rho}|SJ\rangle=\int\frac{\ud^3\vec{p}}{(2\pi)^3}\sqrt{p^2+m^2}(n_{SJ}+\frac{1}{2}),
\eeq
where
\beq
n_{SJ}=\frac{1}{4}\frac{e^{-i\pi\lambda}}{\sqrt{\frac{p^2}{m^2}+1}}\left\{\left(\frac{p^2}{m^2}+1\right)\left|W_{\lambda,\mu}(z)\right|^2-\frac{1}{z^2}\left|2z\frac{dW_{\lambda,\mu}(z)}{dz}-\frac{3}{2}W_{\lambda,\mu}(z)\right|^2\right\}-\frac{1}{2}.
\eeq
The variables used above are defined as follows:
\beq
\vec{p}=\frac{\vec{k}}{a},\qquad
H=\frac{a'}{a^2}=\frac{1}{\alpha\eta^2},\qquad
z=i\tilde{m}\eta^2=i\frac{m}{H},\qquad
\lambda=\frac{-ik^2}{4\tilde{m}}=\frac{-ip^2}{4mH}.
\eeq
As before $\mu=\frac{1}{4}$, and $p=|\vec{p}|$.
In particular,
$H$ is the Hubble parameter and $\vec{p}$ is the physical momentum of the
Fourier mode with comoving wavenumber $\vec{k}$.

For a thermal bath of relativistic bosons at temperature $T$,
the energy density takes the form
$\rho = \int\frac{\ud^3\vec{p}}{(2\pi)^3} \, E \, n_{BE}$,
where
$E=\sqrt{p^2+m^2}$
and
$n_{BE}=\frac{1}{e^{E/T}-1}$
is the Bose-Einstein distribution.
This relation can be inverted to get
$T=\frac{E}{\ln\left(1+\frac{1}{n_{BE}}\right)}$.
In order to see how ``close to thermal''
our state
is, we similarly define
the
``effective temperature'' of a mode as
\beq
  T_{SJ}(p)=\frac{\sqrt{p^2+m^2}}{\ln\left(1+\frac{1}{n_{SJ}}\right)} \ .
\eeq
The more constant $T_{SJ}$ is as a function of $p$, the closer the
distribution $n_{SJ}$ is to being thermal.
Here we define
$n_{SJ}$ to include only excitations \textit{above}
the state $|GS\rangle$ that
minimizes the Hamiltonian at a particular instant of time,
and for which
$\langle GS|\hat{\rho}|GS\rangle=\int\frac{\ud^3\vec{p}}{(2\pi)^3}\frac{1}{2}\sqrt{p^2+m^2}$.
Then
$\langle SJ|\hat{\rho}|SJ\rangle-\langle GS|\hat{\rho}|GS\rangle = \int\frac{\ud^3\vec{p}}{(2\pi)^3}\sqrt{p^2+m^2} \, n_{SJ}$.

Figure \ref{fig2} shows the behaviour of $T_{SJ}$ for different ratios of $\frac{m}{H}$ and $\frac{p}{m}$.
\begin{figure}[h!]
\centering
\includegraphics[bb=25 25 595.3  241.9, width=1.5\textwidth]{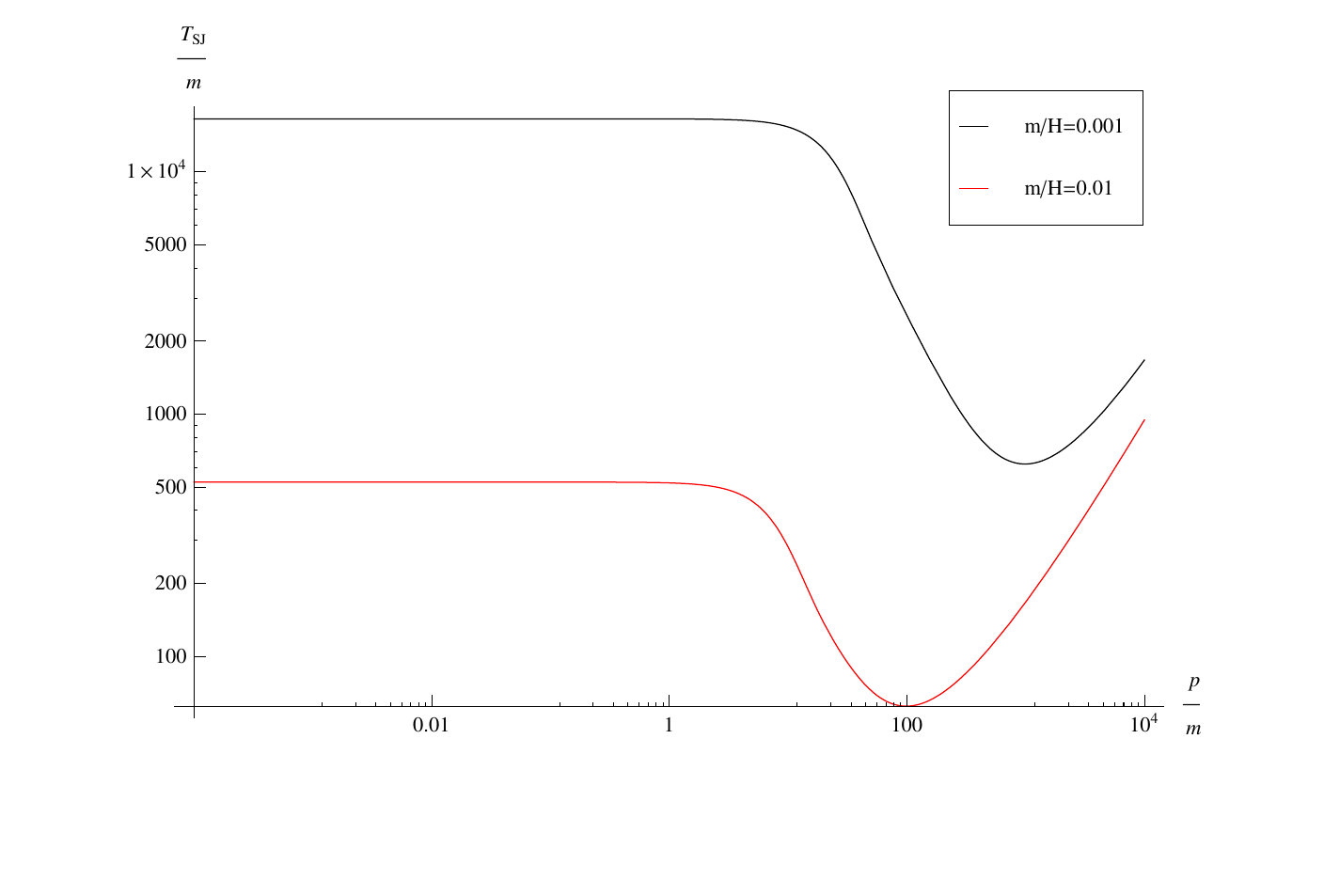}
\includegraphics[bb=15 25 595.3  241.9, width=1.6\textwidth]{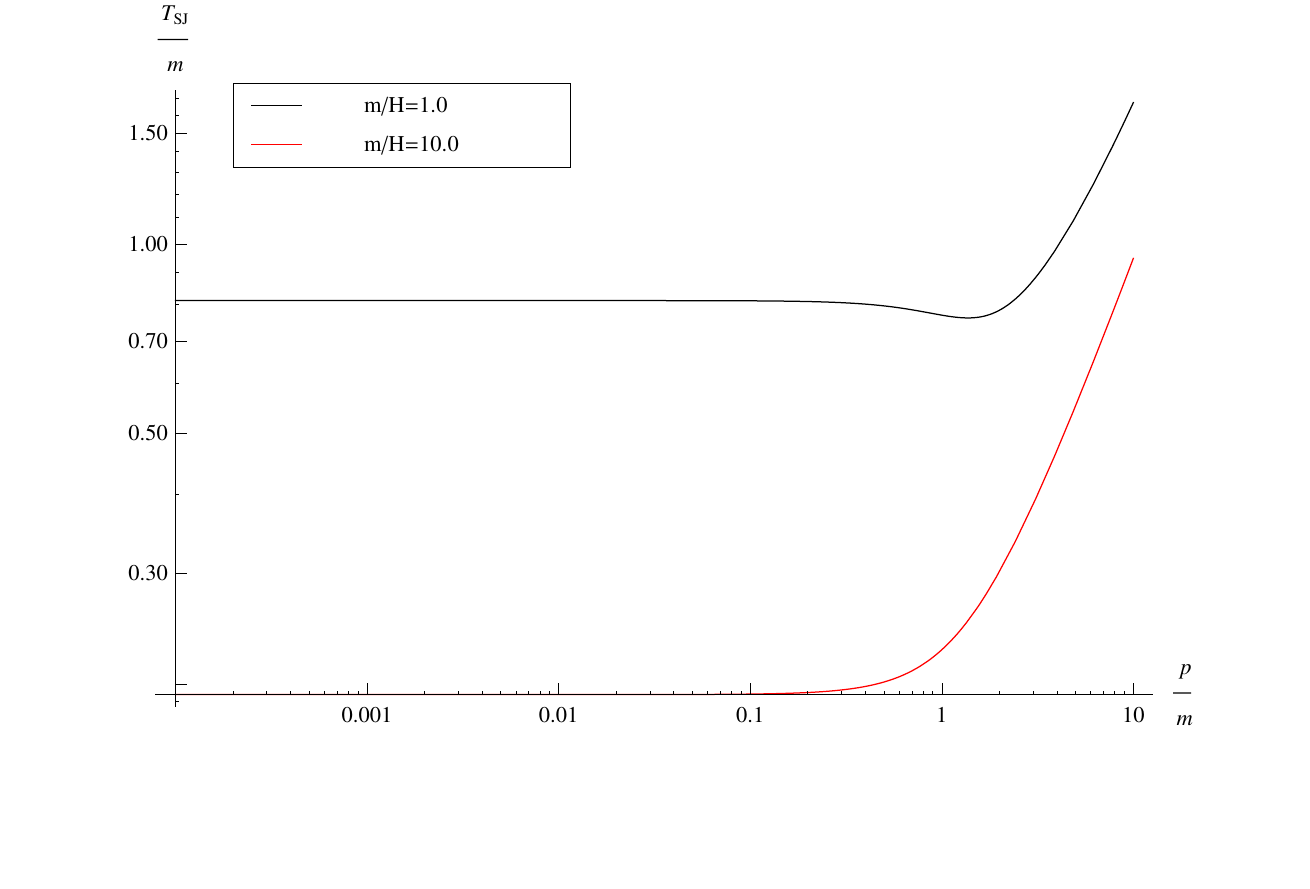}
\caption{Behaviour of $T_{SJ}$ for different ratios of $\frac{m}{H}$ and $\frac{p}{H}$.}
\label{fig2}
\end{figure}
It is evident that the long wavelength modes are in fact at a constant
effective
temperature.
For example,
in the regime where $m\ll H$ and $p\ll\sqrt{mH}$, the Whittaker function has a
simple asymptotic expansion
\beq
 W_{\lambda,\frac{1}{4}}(z)\longrightarrow\frac{\sqrt{\pi}}
 {\Gamma(\frac{3}{4}-\lambda)}z^{\frac{1}{4}}-\frac{2\sqrt{\pi}}{\Gamma(\frac{1}{4}-\lambda)}z^{3/2}+\mathcal{O}(z^{5/4}),
\eeq
using which it can be shown that
\beq
 T_{SJ}(p)\sim\frac{\pi}{4\left|\Gamma(3/4)\right|^2}\frac{H^{3/2}}{\sqrt{m}}.
\eeq
This result suggests that there are correlations on super-horizon
scales.
It is noteworthy that these correlations have appeared without the help
of any previous epoch of accelerated expansion.  Potentially, they could
therefore open up a new perspective on the question of primordial
fluctuations and on the related puzzle sometimes called ``horizon problem''.

\section{Causality and the S-J Vacuum}
Like other vacuua, the S-J vacuum is defined globally, and it
depends on both the causal past
\textit{and} future of the spacetime.
Consider for example a spacetime
which is first static, then expands for a short time, then goes back to
being static again.
\footnote{Of course, such a spacetime is not necessarily a solution to the Einstein equations.}
In light of the inherent time-reversal symmetry of
the conditions defining our vacuum-state, it is clear that this state
can agree neither with the early-time vacuum
(the state of minimum energy at that time),
nor with the late-time vacuum.  Rather, it must strike some sort of
``compromise'' between them.

In the present section we will illustrate this behavior with a simple
example, but before doing so, we would like to dwell for a bit on the
question of whether one should interpret this type of dependence on the
future as a failure of causality.
%
By construction, our definition of the vacuum depends on the full
spacetime geometry.  That it thereby fails to be what John Bell called
``locally causal'' is no surprise because, as is well understood by now,
any reasonable quantum state must incorporate nonlocal correlations and
entanglement.  Certainly the Minkowski vacuum does so.  But does this
type of nonlocality also imply genuine acausality?

The prior question that begs for an answer here is what is meant by
acausality in the context of quantum field theory, considering also that
quantum field theory must ultimately find its place within a theory of
full quantum gravity.  If we remain within the ``operationalist''
framework of external agents, ``measurements'' and state-vector
collapse, then causality (in the sense of relativistic causality)
reduces to the impossibility of superluminal signalling.  In this sense,
there is no question of acausality as long as the twin conditions of
spacelike commutativity and hyperbolicity of the field equations are
respected, which by construction they are in the field theory we are
working with in this paper.\footnote%
{The theory of \cite{Johnston} retains spacelike commutativity, but
 hyperbolicity becomes, together with the notion of field-equation
 itself, approximate at best.}
On the other hand, if we try to adopt a more ``objective'' framework
which dispenses with external agents, then we seem to be left without
any clear definition of relativistic causality at all.  That is, we lack
an {\it{}intrinsic} criterion which could decide whether or not physical
influences are propagating outside the light cone or ``into the past''.
But without such a criterion, the meaning of relativistic causality in
general is called into question.

%
%
A further observation also seems relevant here, even if it does not turn
out to be decisive.  Namely, the assumption we have made of a fixed,
non-dynamical spacetime is already ``anticausal'' in a certain sense.
In a full quantum gravity theory the future geometry must evolve
together with, and in mutual dependence on the future matter-field.
Hence, any
attempt to specify the geometry in advance amounts to imposing a future
boundary condition on the combined system of metric plus scalar field.
Given this, it would not be surprising if a correct semiclassical
treatment of the scalar were also to involve some degree of ``dependence on the
future''.

The specific model we will consider is a 1+1 dimensional
FLRW universe with
metric $ds^2=C(\eta)(d\eta^2-dx^2)$, where
$C(\eta)=A+B\tanh(\rho\eta)$. In the infinite past $C(\eta)\to A-B$ and
in the infinite future $C(\eta)\to A+B$.  It
is
known that there
are normalized modes $u_k^{in}(\eta,x)$ that behave like
positive
frequency Minkowski-space modes in the remote past ($\eta,t\to-\infty$):
\footnote{See section 3.4 of \cite{BD}.}
\bea
u_k^{in}(\eta,x)&=&\frac{1}{\sqrt{4\pi\omega^{in}_k}}e^{ikx-i\omega^{+}_k\eta-(i\omega^{-}_k/\rho)\ln[2\cosh(\rho\eta)]} \notag\\
&\times& {}_2F_1(1+(i\omega^{-}_k/\rho),i\omega^{-}_k/\rho;1-(i\omega^{in}_k/\rho);\frac{1}{2}(1+\tanh(\rho\eta)))\notag\\
&\stackrel{\eta\to-\infty}{\longrightarrow}&\frac{1}{\sqrt{4\pi\omega^{in}_k}}e^{ikx-i\omega^{in}_k\eta},
\eea
where ${}_2F_1$ is the ordinary hypergeometric function and
\bea
\omega^{in}_k&=&[k^2+m^2(A-B)]^{1/2}\notag\\
\omega^{out}_k&=&[k^2+m^2(A+B)]^{1/2}\notag\\
\omega^{\pm}_k&=&\frac{1}{2}(\omega^{out}_k\pm\omega^{in}_k).
\eea
Similarly, there are normalized modes $u_k^{out}(\eta,x)$ that behave
like the positive frequency Minkowski-space modes in the remote future:
($\eta,t\to\infty$)
\bea
u_k^{out}(\eta,x)&=&\frac{1}{\sqrt{4\pi\omega^{out}_k}}e^{ikx-i\omega^{+}_k\eta-(i\omega^{-}_k/\rho)\ln[2\cosh(\rho\eta)]} \notag\\
&\times& {}_2F_1(1+(i\omega^{-}_k/\rho),i\omega^{-}_k/\rho;1+(i\omega^{out}_k/\rho);\frac{1}{2}(1-\tanh(\rho\eta)))\notag\\
&\stackrel{\eta\to\infty}{\longrightarrow}&\frac{1}{\sqrt{4\pi\omega^{out}_k}}e^{ikx-i\omega^{out}_k\eta}.
\eea
The in and out modes are related to eachother by the following Bogolubov transformation
\beq
u_k^{in}(\eta,x)=\alpha_ku_k^{out}(\eta,x)+\beta_k\bar{u}_{-k}^{out},
\eeq where
\bea
\alpha_k&=&(\frac{\omega^{out}_k}{\omega^{in}_k})^{1/2}\frac{\Gamma(1-(i\omega^{in}_k/\rho))\Gamma(-(i\omega^{out}_k/\rho))}{\Gamma(1-(i\omega^{+}_k/\rho))\Gamma(-i\omega^{+}_k/\rho)}\\
\beta_k&=&(\frac{\omega^{out}_k}{\omega^{in}_k})^{1/2}\frac{\Gamma(1-(i\omega^{in}_k/\rho))\Gamma(-(i\omega^{out}_k/\rho))}{\Gamma(1+(i\omega^{-}_k/\rho))\Gamma(i\omega^{-}_k/\rho)}.
\eea
The modes
$u^{in}_k$ and $u^{out}_k$
define
 vacuum states at early and late times, respectively. If the
system is at first ($\eta\to-\infty$) in the in-vacuum state, i.e. the
no particle state, it will have $|\beta_k|^2$ particles of momentum $k$
with respect to the out-vacuum after the expansion ($\eta\to\infty$).
The S-J vacuum has a different nature, simply because the
vacuum state in the $\eta\to-\infty$ region depends on what happens in
the infinite future (and vice-versa).
We can find the S-J vacuum by substituting the modefunctions $u_k^{in}(\eta,x)$
in formula \eqref{VAC}. Defining $g_k(\eta)$ through
$u_k^{in}(\eta,x)=\frac{e^{ikx}}{\sqrt{2\pi}}g_k(\eta)$,
it can be easily verified that
\bea
<u_k^{in},u_{k'}^{in}>&=&\delta(k-k')\int_{-\Lambda}^{\Lambda}\bar{g}_{k}(\eta)g_{k'}(\eta)C(\eta)\ud\eta\\
<\bar{u}_k^{in},u_{k'}^{in}>&=&\delta(k+k')\int_{-\Lambda}^{\Lambda}g_{k}(\eta)g_{k'}(\eta)C(\eta)\ud\eta,
\eea where as usual, we have regulated the integrals with a cutoff $\Lambda$. The asymptotic behaviour of $g_k(\eta)$ is given by
\bea
g_k(\eta)&\stackrel{\eta\to-\infty}{\longrightarrow}&\frac{1}{\sqrt{2\omega^{in}_k}}e^{-i\omega^{in}_k\eta}\\
g_k(\eta)&\stackrel{\eta\to\infty}{\longrightarrow}&\frac{1}{\sqrt{2\omega^{out}_k}}(\alpha_ke^{-i\omega^{out}_k\eta}+\beta_ke^{i\omega^{out}_k\eta}).
\eea
Using these expressions, it can be checked that
\beq
\lim_{\Lambda\to\infty}\frac{<\bar{u}_{-k}^{in},u_{k}^{in}>}{<u_k^{in},u_{k}^{in}>}\longrightarrow\frac{2\alpha_k\beta_k}{|\alpha_k|^2+|\beta_k|^2+\frac{\omega_{out}}{\omega_{in}}\frac{A-B}{A+B}}\equiv\gamma_k,
\eeq from which the S-J vacuum can be computed using \eqref{VAC}:
\bea
u^{SJ}_k(\eta,x)&=&\mu_ku^{in}_{k}(\eta,x)+\xi_k\bar{u}^{in}_{-k}(\eta,x)\\
\mu_k&=&\frac{1}{\sqrt{2}}\left[1+\frac{1}{\sqrt{1-|\gamma_k|^2}}\right]^{\frac{1}{2}}\\
\xi_k&=&-\frac{1}{\sqrt{2}}\left[-1+\frac{1}{\sqrt{1-|\gamma_k|^2}}\right]^{\frac{1}{2}}e^{i\text{Arg}(\alpha_k\beta_k)}.
\label{deviation}
\eea
Fig. \ref{fig1} shows the difference between the S-J and ``in" vacuua
for a specific set of parameters and frequencies.
As
one would expect,
this deviation is only significant for low-frequency modes, which are more sensitive to the rate of expansion $\rho$.

\begin{figure}[h!]
\centering
\includegraphics[bb=0 25 595.3  400.0, width=0.9\textwidth]{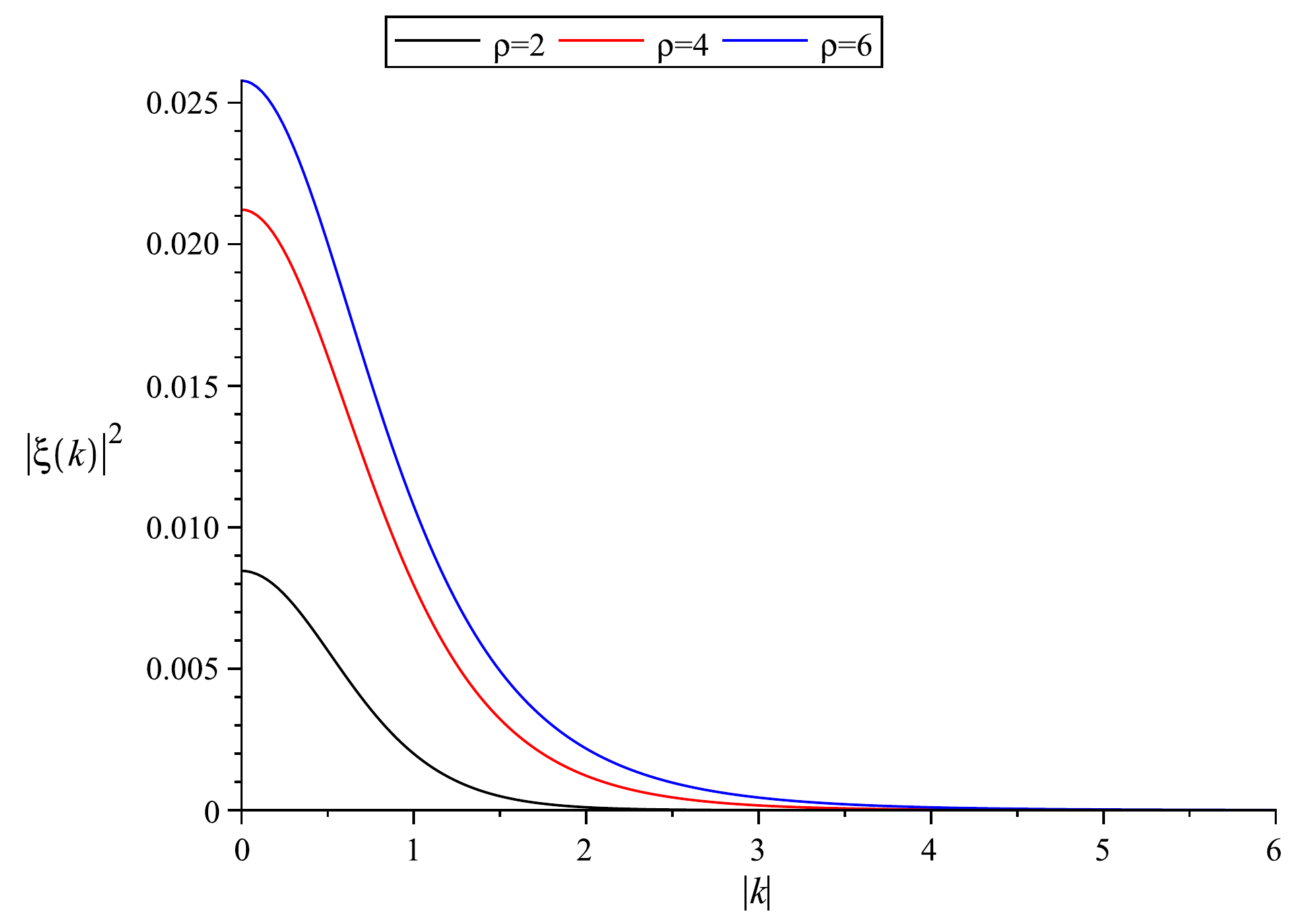}
\caption{The deviation between the S-J and ``in" vacuua, as
measured
by $|\xi(k)|^2$, defined in \eqref{deviation}.
The parameters used here are $A=2.0$, $B=1.0$, and $m=1.0$.}
\label{fig1}
\end{figure}
\section{Conclusions and Discussions}
We have defined a distinguished vacuum for a free quantum field in a
globally hyperbolic region of an arbitrarily curved spacetime.  This
``S-J'' state is well-defined for all compact regions and for a large
class of noncompact ones.

We have shown that for static spacetimes, our vacuum
coincides with the usual ground state.
We have also computed it explicitly for a scalar field of mass $m\ge 0$
in a radiation-filled, spatially flat, homogeneous and isotropic cosmos.
In that connection we also computed an ``effective temperature that can
be defined for the super-horizon modes of the massive field.  The
correlations found thereby could open up a new perspective on the
question of primordial fluctuations and the so-called ``horizon
problem".

A peculiar aspect of our prescription is its temporal non-locality.
We demonstrated this feature by the example of a spacetime which sandwiches a
region with curvature in-between flat initial and final regions,
but we did not explore its phenomenological implications any further.
In a parallel effort \cite{deSitter}, we have
also
applied
our prescription to de Sitter space, obtaining results for both
the full spacetime and for the so-called
Poincar{\'e} patch.
The vacua obtained thereby differ from the Euclidean (Bunch-Davies)
vacuum below a certain mass threshold, with potentially
interesting phenomenology.

A question that we have not addressed in this paper is whether, or in
which circumstances, the S-J vacuum obeys the so-called Hadamard
condition. Since this work was completed, some results  have appeared
\cite{FV} showing that the answer is yes in some cases and no
in others.  
We hope to return to this and related matters elsewhere.

\acknowledgments

We would like to thank Yasaman K. Yazdi, Fay Dowker, Michel Buck, and Bill Unruh
for useful discussions and comments throughout the
course of this project.  We are supported by the University of Waterloo and the
Perimeter Institute for Theoretical Physics.  Research at the Perimeter
Institute is supported by the Government of Canada through Industry
Canada and by the Province of Ontario through the Ministry of Research
\& Innovation.

\appendix
\section{Comments and Calculations}
\label{comments}
\subsection{When is $i\Delta$ self-adjoint?}
\label{com1}
We saw in the main text that $i\Delta$ is self-adjoint and has a
complete set of normalizable eigenvectors
when $i\Delta(x,y)\in L^2(M\times M)$.
In 3+1 dimensions, this Hilbert-Schmidt condition is not satisfied,
however, because the retarded Green's function is a distribution, and
not a function.
Nonetheless, it is
still possible to show that $i\Delta$ is self-adjoint
when $M$ is bounded.
Here, we will
show that this is indeed the case for a bounded region $M$ of Minkowski
space, and will argue that this conclusion should continue to hold in all
curved spacetimes.

In $3+1$ dimensional
Minkowski space ($\mathbb{M}^4$),
we have $G_R(x,y)=G(x-y)$ with $G(x)$ given by
\beq
G(x)=\theta(t)\theta(s^2)\Big(\frac{\delta(s^2)}{2\pi}-\frac{m}{4\pi s}J_1(ms)\Big),
\label{GrF}
\eeq
where $s=\sqrt{t^2-\vec{x}^2}$ for $t^2\ge\vec{x}^2$, and
$s=-i\sqrt{\vec{x}^2-t^2}$ for $t^2\le\vec{x}^2$. $J_1$ is a Bessel
function of
the
first kind and $m$ is the mass of the scalar field. We will
now show that $i\Delta$ is a bounded operator on $\Co$, i.e. we will
prove that there exists $N>0$ such that $||i\Delta f||\le N||f||$ for
all $f\in\Co$, where $||\cdot||$ is the $L^2(\mathbb{M}^4)$
norm. \footnote{$||\psi||=\sqrt{<\psi,\psi>}$.}
Let $G_R f=\int_MG_R(x,y)f(y)\ud V_y$ and $G_A f=\int_MG_A(x,y)f(y)\ud
V_y$ denote the retarded and advanced solutions of the Klein-Gordon
equation with source $f$, respectively. It is enough to show that $G_R$
and $G_A$ are bounded because $||i\Delta f||=||iG_R f-iG_A f||\le||G_R
f||+||G_A f||$.

Let $A(x)=\theta(t)\theta(s^2)\frac{\delta(s^2)}{2\pi}$ and
$B(x)=\theta(t)\theta(s^2)\frac{m}{4\pi s}J_1(ms)$ so that
$G(x)=A(x)-B(x)$.  Then: $||G_R
f||=||Af-Bf||\le||Af||+||Bf||\le||Af||+\frac{m^2V}{8\pi}||f||$, where V
is the total spacetime volume. \footnote{B is Hilbert-Schmidt because
 $\int_M\int_M|B(x-y)|^2\ud V_x\ud V_y\le(\frac{m^2V}{8\pi})^2$, since
 $\frac{J_1(ms)}{s}$ peaks at $\frac{m}{2}$.}  Therefore, $G_R$ is
bounded if and only if $A$ is bounded. Consider a bounded region of
$\mathbb{M}^4$ where $\vec{x}\in[\vec{x}_{min},\vec{x}_{max}]$ and
$t\in[t_{min},t_{max}]$, and a smooth function of compact support
$f(t,\vec{x})$ on this region.
It can be shown that:
\beq
(Af)(t,\vec{x})=\int_{\vec{x}_{min}}^{\vec{x}_{max}}\frac{f(t-|\vec{x}-\vec{y}|,\vec{y})}{4\pi|\vec{x}-\vec{y}|}\ud^3\vec{y}.
\eeq It then follows from the Cauchy-Schwarz inequality that
\beq
|(Af)(t,\vec{x})|^2\le\frac{1}{16\pi^2}
\Big\{\int_{\vec{x}_{min}}^{\vec{x}_{max}}\frac{\ud^3\vec{y}}{|\vec{x}-\vec{y}|^2}\Big\}\times
\Big\{\int_{\vec{x}_{min}}^{\vec{x}_{max}}|f(t-|\vec{x}-\vec{y}|,\vec{y})|^2\ud^3\vec{y}\Big\}.
\eeq
Also,
\beq
\int_{\vec{x}_{min}}^{\vec{x}_{max}}\frac{\ud^3\vec{y}}{|\vec{x}-\vec{y}|^2}
=\int_{\vec{x}_{min}+\vec{x}}^{\vec{x}_{max}+\vec{x}}\frac{\ud^3\vec{y}}{|\vec{y}|^2}\le
\int_{2\vec{x}_{min}}^{2\vec{x}_{max}}\frac{\ud^3\vec{y}}{|\vec{y}|^2}\equiv C_0,
\eeq where $C_0$ is some finite positive number. So we have that
\beq
|(Af)(t,\vec{x})|^2\le\frac{C_0}{16\pi^2}
\Big\{\int_{\vec{x}_{min}}^{\vec{x}_{max}}|f(t-|\vec{x}-\vec{y}|,\vec{y})|^2\ud^3\vec{y}\Big\}.
\eeq
It then follows that
\bea
||Af||^2&=&\int_{\vec{x}_{min}}^{\vec{x}_{max}}\ud^3\vec{x}\int_{t_{min}}^{t_{max}}\ud t |(Af)(t,\vec{x})|^2 \\
&\le& \frac{C_0}{16\pi^2}\int_{\vec{x}_{min}}^{\vec{x}_{max}}\ud^3\vec{x}\int_{\vec{x}_{min}}^{\vec{x}_{max}}\ud^3\vec{y}
\int_{t_{min}}^{t_{max}}\ud t |f(t-|\vec{x}-\vec{y}|,\vec{y})|^2\\
&=& \frac{C_0}{16\pi^2}\int_{\vec{x}_{min}}^{\vec{x}_{max}}\ud^3\vec{x}\int_{\vec{x}_{min}}^{\vec{x}_{max}}\ud^3\vec{y}
\int_{t_{min}+|\vec{x}-\vec{y}|}^{t_{max}+|\vec{x}-\vec{y}|}\ud t |f(t,\vec{y})|^2.
\eea
Since $f(t,\vec{y})$ vanishes for $t>t_{max}$:
$\int_{t_{min}+|\vec{x}-\vec{y}|}^{t_{max}+|\vec{x}-\vec{y}|}\ud t |f(t,\vec{y})|^2
\le\int_{t_{min}}^{t_{max}}\ud t |f(t,\vec{y})|^2$, from which it follows that
\beq
||Af||^2\le \frac{C_0V_s}{16\pi^2}||f||^2,
\eeq where $V_s=\int_{\vec{x}_{min}}^{\vec{x}_{max}}\ud^3\vec{x}$ is the enclosed spatial volume.
A similar analysis goes through for the advanced solution $G_Af$ which results in the same bound. At long last:
\beq
||i\Delta f||\le \frac{V}{4\pi}\Big(m^2+2\sqrt{\frac{C_0V_s}{V^2}}\Big)||f||.
\eeq
Because $i\Delta$ is bounded on $\Co$ and $\Co$ is a dense subspace of
$L^2({M})$, $i\Delta$ can be uniquely extended to
$L^2({M})$ as a bounded operator \cite{RS}. It then follows
that $i\Delta$ is self-adjoint because it is Hermitian on all of
$L^2(M)$ \cite{RS}.


It should be clear that the boundedness of $i\Delta$ has everything to
do with the singularity structure of $i\Delta(x,y)$ (once we restrict
ourselves to bounded spacetimes). It is not terribly unrealistic to
assume that this singularity structure remains (more or less) the same
in curved spacetimes. This is certainly true in the coincidence limit,
if the equivalence principle is respected. Based on these arguments, we
assume
in this paper
 that $i\Delta$ is a self-adjoint operator on $L^2(M)$, for all
bounded globally-hyperbolic spacetimes $M$.

\subsection{The S-J Vacuum and the Simple Harmonic Oscillator}
\label{com3}
This simple example
illustrates
the technical difficulties one faces when diagonalizing $i\Delta$, and how they can be resolved.
Consider a simple harmonic oscillator
with unit mass and frequency $\omega$, whose position $q(t)$ satisfies
$(\frac{d^2}{dt^2}+\omega^2)q(t)=0$.
The associated retarded Green's function satisfies
$(\frac{d^2}{dt^2}+\omega^2)G_R(t,t')=-\delta(t-t')$,
with $G_R(t,t')=0$ for $t<t'$. The solution
to this equation is $G_R(t,t')=-\theta(t-t')\frac{\sin[\omega(t-t')]}{\omega}$, which in turn gives
\bea
\Delta(t,t')&=&G_R(t,t')-G_R(t',t)=-\frac{\sin[\omega(t-t')]}{\omega} \\
&=&\frac{1}{2i\omega}\big[e^{-i\omega(t-t')}-e^{i\omega(t-t')}\big].
\eea
Taking $t\in(-\infty,\infty)$, it may be verified that,
\footnote{Here, as always, $(\Delta f)(t)=\int\Delta(t,t')f(t')\ud t'$.}
formally,
\beq
i\Delta e^{\pm i\omega t}=\frac{\mp\delta(0)}{2\omega}e^{\pm i\omega t}.
\eeq
Keeping the $\delta(0)$'s around,
we see that
$T_{\mp}(t)=\frac{e^{\pm i\omega t}}{\sqrt{\delta(0)}}$ are orthonormal
eigenfunctions of $i\Delta$ with eigenvalues
$\lambda_{\pm}=\pm\frac{\delta(0)}{2\omega}$.  According to our
prescription, the resulting positive frequency modefunction
$\mathcal{U}^{SJ}(t)=\sqrt{\lambda_{+}}T_{+}(t)=\frac{e^{-i\omega
  t}}{\sqrt{2\omega}}$, which is completely well defined and gives the
right vacuum state:
the state $|0>$ annihilated by $\hat{a}$ which
multiplies $\mathcal{U}^{SJ}(t)$ in the position operator expansion
$\hat{q}(t)=\mathcal{U}^{SJ}(t)\hat{a}+\bar{\mathcal{U}}^{SJ}(t)\hat{a}^\dagger$
is in fact the minimum energy state of the Hamiltonian.
Thus
the infinities appearing in the spectrum of $i\Delta$ end up being
harmless.
In other words, the S-J vacuum state
should
not
depend on how $\delta(0)$ is regularized.

One regularization scheme, for example, is to first restrict
to $t\in[-T,T]$, diagonalize $i\Delta$, and then take the limit
$T\to\infty$ once the spectrum of $i\Delta$
has been
computed.  Let
$u(t)=\frac{e^{-i\omega t}}{\sqrt{2\omega}}$ so that
$i\Delta(t,t')=u(t)\bar{u}(t')-u(t)\bar{u}(t')$. Finding the spectrum of
$i\Delta$ in this case is similar to diagonalizing a two-by-two matrix.
It may be confirmed that there are two eigenvalues $\lambda$ and
$-\lambda$ with corresponding eigenfunctions $T_{+}$ and $\bar{T}_{+}$,
where
\beq
 T_{+}=\sqrt{\frac{<u,u>+\lambda}{2\lambda^2}}u(t)-\sqrt{\frac{<u,u>-\lambda}{2\lambda^2}}e^{-i\text{arg}(<u,\bar{u}>)}\bar{u}(t)
\eeq
\beq
 \lambda=\sqrt{{<u,u>}^2-|<u,\bar{u}>|^2}.
\eeq
These expressions are completely well-defined because all
the
inner products are finite.
The corresponding positive frequency modefunction as dictated by our prescription is then
\beq
 \mathcal{U}^{SJ}(t)=\sqrt{\lambda}T_{+}(t)=\frac{1}{\sqrt{2}}\left\{\sqrt{1+\frac{1}{\sqrt{1-\frac{|<u,\bar{u}>|^2}{<u,u>^2}}}} u(t)-
 \sqrt{\frac{1}{\sqrt{1-\frac{|<u,\bar{u}>|^2}{<u,u>^2}}}-1}e^{-i\text{arg}(<u,\bar{u}>)}\bar{u}(t)\right\}.
 \eeq
In
the limit
 $T\to\infty$, the ratio
$\frac{|<u,\bar{u}>|^2}{<u,u>^2}\to0$ and we recover
$\mathcal{U}^{SJ}(t)=\frac{e^{-i\omega t}}{\sqrt{2\omega}}$.


\subsection{Equation \eqref{E1} as an equality between distributions}
\label{com2}
Let us
show that
the
right and left hand sides
of \eqref{E1}
are equal if they are integrated
against
a smooth
test function $f\in\Co$.

Since $\Delta f\in\Sol$ for any $f\in\Co$, we
can expand out $i\Delta f$ in terms of
the $u_k$:
$i\Delta f(x)=\sum_k\alpha_ku_k(x)+\beta_k\overline{u}_k(x)$,
where $\alpha_k$'s
and $\beta_k$'s are constants. It can be verified that
$\alpha_k=(u_k,i\Delta f)_{KG}$ and $\beta_k=-(\overline{u}_k,i\Delta f)_{KG}$.
Then,
\bea
 i\Delta f(x)&=&\sum_k(u_k,i\Delta f)_{KG}u_k(x)-(\overline{u}_k,i\Delta f)_{KG}\overline{u}_k(x)\\
 &=& \sum_k\langle u_k,f\rangle u_k(x)-\langle\overline{u}_k,f\rangle\overline{u}_k(x) \\
 &=& \int_M\left[\sum_{k}u_k(x)\overline{u}_k(y)-\overline{u}_k(x)u_k(y)\right]f(y)\ud V_y. 
\eea

\subsection{Another Slant on the SJ Prescription}
\label{com4}
The following simple but useful result helps put our construction in context.
\begin{thm}
Let $\SolC$ denote the set of all complex solutions of the Klein-Gordon equation.
Assume there are functions $u_i(x)\in\SolC$, which together with their complex conjugates span $\SolC$ and
satisfy
\beq
\label{IPS1}
 (u_n,u_m)_{KG}=\delta_{nm},  \text{ }
 (\bar{u}_n,\bar{u}_m)_{KG}=-\delta_{nm}, \text{ }
 (u_n,\bar{u}_m)_{KG}=0
\eeq
\beq
\label{IPS2}
\frac{<u_n,u_m>}{<u_n,u_n>}=\delta_{nm},\text{ } <u_n,\bar{u}_m>=0,
\eeq
where $<,>:\SolC\times\SolC\to\mathbb{C}$
%
%
is the $L^2$ inner product defined in \eqref{L2}. (See Section \ref{Background} for the definition of $(,)_{KG}$.)
Then, any other
functions $g_n(x)\in\SolC$ that also satisfy
these conditions are unitarily-related to the
$u_i(x)$
(i.e. $g_n(x)=\sum_k\alpha_{nk}u_k(x)$ where
$\sum_k\bar{\alpha}_{nk}\alpha_{mk}=\delta_{nm}$).
\end{thm}
\begin{proof}
Suppose both $u_n(x)$ and $g_n(x)$ satisfy the above conditions and let
\beq
g_n(x)=\sum_k\alpha_{nk}u_k(x)+\beta_{nk}\bar{u}_k(x).
\label{g-u}
\eeq
Requiring
$(g_n,g_m)_{KG}=\delta_{nm}$
implies
\beq
\sum_k\bar{\alpha}_{nk}\alpha_{mk}-\bar{\beta}_{nk}\beta_{mk}=\delta_{nm} \label{temp3}.
\eeq
Also, it follows from \eqref{g-u} and \eqref{IPS1} that
\beq
\alpha_{nk}=(u_k,g_n)_{KG}, \qquad
\beta_{nk}=-(\bar{u}_k,g_n)_{KG}.
\label{alphbeta}
\eeq
Given that $\sum_k|g_k)_{KG}(g_k| - |\overline{g}_k)_{KG}(\overline{g}_k|$
is the identity operator on $\SolC$, it follows from \eqref{alphbeta} that
\bea
  u_n(x)& = &\sum_k(g_k,u_n)_{KG}g_k(x) - (\bar{g}_k,u_n)_{KG} \bar{g}_k(x) \notag \\
  &=&\sum_k\bar{\alpha}_{kn}g_k(x) - \beta_{kn}\bar{g}_k(x) \ .
\label{u-g}
\eea
Using \eqref{g-u} and \eqref{IPS2}, it can
also
be verified that
\beq
  \beta_{nk} = \frac{<\bar{u}_k,g_n>} {<u_k,u_k>} \ .
  \label{beta1}
\eeq
Similarly,
\eqref{u-g} and the $L^2$ orthogonality conditions for the $g_i(x)$s
(which amounts to replacing all $u$'s with $g$'s in \eqref{IPS2})
imply
\beq
  \beta_{nk} = -\frac{<\bar{g}_n,u_k>} {<g_n,g_n>} \ .
  \label{beta2}
\eeq
Given that
$\langle\bar{u}_k,g_n\rangle \,=\, \langle\bar{g}_n,u_k\rangle$,
$\langle u_k,u_k \rangle \, >  0$, and
$\langle g_n,g_n \rangle \, > 0$,
it follows from \eqref{beta1} and \eqref{beta2} that
$\langle\bar{u}_k,g_n\rangle\,=\, 0$.
This in turn implies that $\beta_{nk}=0$,
which together with \eqref{temp3} proves the theorem.
\end{proof}
Note that this proof is valid for any Hermitian inner product $<,>$ on
$\SolC$ which enjoys the additional property $<\bar{f},g>\,=\,<\bar{g},f>\,.$

We have shown here that requiring \eqref{IPS2}, in addition to the
usual quantization conditions, is enough to define a distinguished vacuum state.
That these modes are eigenfunctions of $i\Delta$ ($(i\Delta u_n)(x)=<u_n,u_n>u_n(x)$),
follows from \eqref{E1}.

\section{Useful Theorems}
\label{theorems}
\begin{thm}\footnote%
 {This is an exact restatement of Theorem 4.1.2 of \cite{Wald}, which we have included for the convenience of the reader.
 The symbol `$D(S)$' denotes the so-called domain of dependence of $S$.}
 Let $(M,g_{ab})$ be a globally hyperbolic spacetime with smooth, spacelike Cauchy surface $\Sigma$. Then the Klein-Gordon
 equation \eqref{KG} has a well posed initial value formulation in the following sense: Given any pair of smooth ($\textrm{C}^{\infty}$)
 functions $(\phi_0,\dot{\phi}_0)$ on $\Sigma$, there exists a unique solution $\phi$ to \eqref{KG}, defined on all of $M$, such that
 on $\Sigma$ we have $\phi=\phi_0$ and $n^a\nabla_a\phi=\dot{\phi}_0$, where $n^a$ denotes the unit (future-directed) normal to $\Sigma$.
 Furthermore, for any closed subset $S\subset\Sigma$, the solution, $\phi$, restricted to D(S) depends only upon the initial data on S.
 In addition, $\phi$ is smooth and varies continuously with the initial data.
 \label{main}
\end{thm}
The above theorem continues to hold if a (fixed, smooth) ``source term"
$f$ is inserted on the right hand side of \eqref{KG}. It then follows
from the domain of dependence feature of the theorem that there exist
unique advanced and retarded solutions to the Klein-Gordon equation with
source, whence in a globally hyperbolic spacetime there exist unique
advanced and retarded Green's functions for the Klein-Gordon equation.

\begin{thm}\footnote{This is a generalization of Lemma 3.2.1 of \cite{Wald}.}
 $\Delta:\textrm{C}_0^{\infty}\to\Sol$ satisfies $\int_M\phi f\ud V=\Omega(\Delta f,\phi)$
 for all $f\in\textrm{C}_0^{\infty}$ and $\phi\in\Sol$~.
 \label{main2}
\end{thm}
\begin{proof}
Let's first give a simple proof for 1+1 dimensional Minkowski space and
then generalize the arguments to globally hyperbolic spacetimes of
higher spatial dimension.
Pick $t_1,t_2\in\Reals$ such that $f=0$ for $t\notin[t_1,t_2]$.
Then:
\beq
\int_M\phi f\ud V=\int_{t_1}^{t_2}\ud t\int_{-\infty}^{\infty}\phi(\partial_t^2-\partial_x^2+m^2)(Af)\ud x,
\label{dummy1}
\eeq
where $Af$ is the advanced solution of the Klein-Gordon equation with
source $f$.
Now integrating by parts twice yields
\bea
\int_{t_1}^{t_2}\phi\partial_t^2(Af)\ud t&=&\phi\partial_t(Af)|^{t_2}_{t_1}-\partial_t\phi(Af)|^{t_2}_{t_1}+\int_{t_1}^{t_2}(Af)\partial_t^2\phi\ud t \\
 &=& \partial_t\phi(Af)|_{t_1}-\phi\partial_t(Af)|_{t_1}+\int_{t_1}^{t_2}(Af)(\partial_x^2-m^2)\phi\ud t,
\eea where we have used the Klein-Gordon equation and that $Af$ vanishes outside the causal past of the support of $f$.
Substituting this back into \eqref{dummy1}, we can further simplify by noting that
$\int_{-\infty}^{\infty}\big[(Af)\partial_x^2\phi-\phi\partial_x^2(Af)\big]\ud
x=0$~.
(Because $Af$ induces initial data of compact support on all
equal-time spatial slices, all surface terms vanish.)
Using this we find:
\beq
\int_M\phi f\ud V=\int_{-\infty}^{\infty}\big[\partial_t\phi(Af)|_{t_1}-\phi\partial_t(Af)|_{t_1}\big]\ud x=\Omega(\Delta f,\phi),
\eeq since $Rf$ (the retarded solution of the Klein-Gordon equation with source $f$) vanishes at $t=t_1$.
To generalize the proof, we
write
\bea
\int_M\phi f\ud V&=&\int_M\phi(\nabla_{\mu}\nabla^{\mu}+m^2)(Af)\ud V \notag\\
&=& \int_M\nabla_{\mu}\{\phi\nabla^{\mu}(Af)\}\ud V-\int_M\nabla_{\mu}\phi\nabla^{\mu}(Af)\ud V+\int_Mm^2\phi(Af)\notag\\
&=& \int_M\nabla_{\mu}\{\phi\nabla^{\mu}(Af)-(Af)\nabla^{\mu}\phi\}\ud V+\int_M(\nabla_{\mu}\nabla^{\mu}+m^2)\phi(Af)\ud V \notag\\
&=&\int_{\partial M}\{\phi n^{\mu}\nabla_{\mu}(Af)-(Af)n^{\mu}\nabla_{\mu}\phi\}\sqrt{-h}\ud^3x,
\eea where we have used the Stoke's theorem in the last line. Here $n^{\mu}$ is the unit normal to the boundary
$\partial M$ and $h$ is the determinant of the induced metric on the
boundary.
Again, because $Af$ induces initial data of compact support on all
equal-time spatial slices and $Af$ vanishes outside the causal past of the support of $f$:
\bea
 \int_{\partial M}\{\phi n^{\mu}\nabla_{\mu}(Af)-(Af)n^{\mu}\nabla_{\mu}\phi\}\sqrt{-h}\ud^3x
 &=&-\int_{\Sigma_{t_1}}\{\phi n^{\mu}\nabla_{\mu}(Af)-(Af)n^{\mu}\nabla_{\mu}\phi\}\sqrt{-h}\ud^3x\notag\\
 &=&\Omega(\Delta f,\phi),
\eea
where we have used in the last line that $Rf$ vanishes at $\Sigma_{t_1}$.
\end{proof}

\begin{thm}
 $G_R(x,z)=G_A(z,x)$.
 \label{main3}
\end{thm}
\begin{proof}
Let $M$ be the region bounded by a pair of Cauchy surfaces, one to the
future of both $x$ and $z$, the other to their past.
Then
\bea
 G_R(x,z)-G_A(z,x)&=&\int_M\big\{G_R(y,z)\delta^4(y-x)-G_A(y,x)\delta^4(y-z)\big\}\ud^4y \notag\\
 &=&\int_M\big\{G_A(y,x)(\Box_y+m^2)G_R(y,z)-G_R(y,z)(\Box_y+m^2)G_A(y,x)\big\}\ud V_y \notag \\
 &=&\int_M\big\{G_A(y,x)\Box_yG_R(y,z)-G_R(y,z)\Box_yG_A(y,x)\big\}\ud V_y\notag \\
 &=&\int_M\nabla_{\mu}\big\{G_A(y,x)\nabla^{\mu}G_R(y,z)-G_R(y,z)\nabla^{\mu}G_A(y,x)\big\}\ud V_y \notag \\
 &=&\int_{\partial M}\big\{G_A(y,x)n^{\mu}\nabla_{\mu}G_R(y,z)-G_R(y,z)n^{\mu}\nabla_{\mu}G_A(y,x)\big\}
 \sqrt{-h}\ud^3y=0\notag
\eea
where we have used the Stoke's theorem in the last line. Here $n^{\mu}$
is the unit normal to the boundary $\partial M$ and $h$ is the
determinant of the induced metric on the boundary.
We now explain why the last expression is identically zero. By
definition, it is only when $y$ is in the causal future of $z$ and the
causal past of $x$ where this expression could be nonzero. However,
because $y$ is being evaluated at the boundary this is never possible in
a globally hyperbolic spacetime.
\end{proof}

\bibliographystyle{JHEP.bst}
\bibliography{Vacuum}

\end{document}